\documentclass[11pt]{article} 
\usepackage[utf8]{inputenc} 

\usepackage{geometry} 
\usepackage{graphicx, amsmath,amsfonts,cite,appendix, amssymb,xcolor,subfig, hyperref, enumerate, mathrsfs, amsthm,url, mathtools} 
\usepackage{booktabs} 
\usepackage{array} 
\usepackage{paralist} 
\usepackage{verbatim} 
\usepackage{subfig} 

\usepackage{bbm}
\usepackage{pgfgantt}
\usepackage{caption}
\usepackage{sidecap}

\usepackage{algorithm,algorithmicx,algpseudocode}
\algloopdefx{NFor}[1]{\textbf{for} #1 \textbf{do}}

\usepackage[font=small,labelfont=bf,labelsep=space]{caption}
\newcommand{\un}{\underline}
\newcommand{\be}{\begin{equation}}
\newcommand{\ee}{\end{equation}}
\newcommand{\ben}{\begin{equation*}}
\newcommand{\een}{\end{equation*}}
\newcommand{\mc}{\mathcal}

\newtheorem{lem}{Lemma}
\newtheorem{applem}{Lemma}[section]

\newtheorem{thm}{Theorem}

\newcommand{\e}{\epsilon}

\newcommand{\abs}[1]{\left \lvert #1\right \rvert}
\newcommand{\norm}[1]{\left \lVert #1\right \rVert}
\newcommand{\expec}{\mathbb{E}}

\title{\vspace*{-1in} Multiprocessor Approximate Message Passing \\with Column-Wise Partitioning\footnote{This document serves as a supporting document for \cite{MaLuBaron2017}}}
\author{{Yanting Ma}\\ {North Carolina State University}
  \\ {\small {yma7@ncsu.edu}}
\and {Yue M. Lu}\\ {Harvard University}
  \\ {\small {yuelu@seas.harvard.edu}}
 \and {Dror Baron} \\ {North Carolina State University}
  \\ {\small{barondror@ncsu.edu}}
}
\date{\vspace{-0.2in}}

\begin{document}
\maketitle
\begin{abstract}
Solving a large-scale regularized linear inverse problem using multiple processors is important in various real-world applications due to the limitations of individual processors and constraints on data sharing policies. This paper focuses on the setting where the matrix is partitioned column-wise. We extend the algorithmic framework and the theoretical analysis of approximate message passing (AMP), an iterative algorithm for solving linear inverse problems, whose asymptotic dynamics are characterized by state evolution (SE). In particular, we show that column-wise multiprocessor
AMP (C-MP-AMP) obeys an SE under the same assumptions when the SE for AMP holds. The SE results imply that ({\em i}) the SE of C-MP-AMP converges to a state that is no worse than that of AMP and ({\em ii}) the asymptotic dynamics of C-MP-AMP and AMP can be identical. Moreover, for a setting that is not covered by SE, numerical results show that damping can improve the convergence performance of C-MP-AMP.
\end{abstract}

\section{Introduction}
\label{sec:intro}

Many scientific and engineering problems can be modeled as
solving a regularized linear inverse problem of the form
\begin{equation}
y = Ax + w, 
\label{eq:SP_sys} 
\end{equation}
where the goal is to estimate the unknown $x\in\mathbb{R}^N$ given the matrix $A\in\mathbb{R}^{n\times N}$ and statistical information about the signal $x$ and the noise $w\in\mathbb{R}^n$. 

In some scenarios, it might be desirable to partition the matrix $A$ either column-wise or row-wise and store the sub-matrices at different processors. 
The partitioning style depends on data availability, computational considerations, and privacy concerns.
For example, in high-dimensional settings where $N\gg n$, 
or in situations where the columns of $A$, which represent features in feature selection problems~\cite{Hastie2001}, cannot be shared among processors for privacy preservation, column-wise partitioning might be preferable.
In this paper, we consider multiprocessor computing for the (non-overlapping) column-wise partitioned linear inverse problem:
\begin{equation}
y = \sum_{p=1}^P A_p x_p + w, 
\label{eq:MP_sys}
\end{equation}
where $P$ is the number of processors, $A_p\in\mathbb{R}^{n\times N_p}$ is the sub-matrix that is stored in Processor $p$, and $\sum_{p=1}^P N_p=N$.

Many studies 
on solving the column-wise partitioned linear inverse problem~(\ref{eq:MP_sys})
have been in the context of distributed feature selection.
Zhou {\em et al.}~\cite{Zhou2014} modeled feature selection as a parallel group testing problem.
Wang {\em et al.}~\cite{WangDunsonLeng2016} proposed to de-correlate the data matrix before partitioning, and each processor then works independently using the de-correlated matrix without communication with other processors.
Peng {\em et al.}~\cite{PengYanYin2013} studied problem~(\ref{eq:MP_sys}) in the context of optimization, where they proposed a greedy coordinate-block descent algorithm and a parallel implementation of 
the fast iterative shrinkage-thresholding algorithm
(FISTA)~\cite{Beck2009FISTA}.

Our work is based on the approximate message passing (AMP) framework~\cite{DMM2009}.
AMP is an efficient iterative algorithm for solving 
linear inverse problems~(\ref{eq:SP_sys}). 
In the large scale random setting, its average asymptotic dynamics are characterized by a state evolution (SE) formalism~\cite{Bayati2011}, which allows one to accurately predict the average estimation error at every iteration. Recently, a finite-sample analysis of AMP \cite{RushV16} showed that when the prior distribution of the input signal $x$ has i.i.d. sub-Gaussian entries,\footnote{A random variable $X$ is sub-Gaussian if there exist positive constants $c$ and $\kappa$ such that $P(|X-\mathbb{E}X|>\epsilon)\leq ce^{-\kappa\epsilon^2}, \forall \epsilon >0$.} the average performance of AMP concentrates to the SE prediction at an exponential rate in the signal dimension $N$.

Our goal is to extend the AMP algorithmic framework and the SE analysis in \cite{RushV16} to the column-wise
partitioned linear inverse problem~(\ref{eq:MP_sys}).
We show that column-wise multiprocessor AMP (C-MP-AMP) obeys a new SE under the same model assumptions where the SE for AMP holds. 
With the new SE, we can predict the average estimation error in each processor at every iteration.
Moreover, the comparison between the SE of AMP and that of C-MP-AMP implies that
({\em i}) the estimation error of C-MP-AMP is no worse than that of AMP and
({\em ii}) with a specific communication schedule between the processors and the fusion center that coordinates the processors, the asymptotic dynamics
of C-MP-AMP are identical to that of AMP. This result implies a speedup linear in the number of processors.

It is worth mentioning that row-wise multiprocessor AMP~\cite{Han2014,Han2015ICASSP,HanZhuNiuBaron2016ICASSP} obeys the same SE as AMP, because it distributes
the computation of matrix-vector multiplication among multiple processors and aggregates the results before any other operations. 
Some existing work on row-wise multiprocessor AMP~\cite{HanZhuNiuBaron2016ICASSP,ZhuBeiramiBaron2016ISIT,ZhuBaronMPAMP2016ArXiv}
introduces lossy compression to the communication between processors and the fusion center, 
whereas we assume perfect communication and focus on the theoretical justifications and implications of the new SE of C-MP-AMP.

The remainder of the paper is organized as follows. 
Section~\ref{sec:algo} introduces the C-MP-AMP algorithm (Algorithm \ref{algo:mp_amp_algo}), the state evolution sequences, and our main performance guarantee (Theorem \ref{thm:mainresult}), which is a concentration result for PL loss functions acting on the outputs generated by Algorithm \ref{algo:mp_amp_algo} concentrates to the state evolution prediction.
Section~\ref{sec:proof} proves Theorem \ref{thm:mainresult}. The proof is mainly based on Lemmas \ref{lem:cond_htbt} and \ref{lem:main_lemma}. The proof of Lemma \ref{lem:cond_htbt} is the same as in \cite{RushV16} using the result that we prove in Lemma \ref{lem:Ap_cond}. Section \ref{sec:proof_mainlem} proves Lemma \ref{lem:main_lemma}.

\section{Column-Wise Multiprocessor AMP and State Evolution}
\label{sec:algo}

\subsection{Review of AMP}
\label{subsec:review_AMP}
Approximate message passing (AMP)~\cite{DMM2009} is a fast iterative algorithm 
for solving linear inverse problems (\ref{eq:SP_sys}).
Starting with an all-zero vector $x^0$ as its initial estimate, at the $t$th iteration, AMP proceeds according to
\begin{align}
z^t &= y - A x^t+\frac{z^{t-1}}{n}\sum_{i=1}^N \eta_{t-1}'([x^{t-1}+A^*z^{t-1}]_i),\label{eq:AMP1}\\
x^{t+1}&=\eta_t(x^t+A^*z^t)\label{eq:AMP2},
\end{align}
where vectors with negative iteration indices are all-zero vectors, $A^*$ denotes
the transpose of a matrix $A$, $\eta_t:\mathbb{R}\rightarrow\mathbb{R}$ is a Lipschitz function with weak derivative $\eta_t'$, 
for any $u\in\mathbb{R}^N$, $[u]_i$ denotes its $i$th entry. The function $\eta_t$ acts coordinate-wise when applied to vectors. That is,
the vector $(\eta_t(u_1),\eta_t(u_2),...,\eta_t(u_N))$ is denoted by $\eta_t(u)$.

Under the assumptions on the measurement matrix $A$, the signal $x$, the measurement noise $w$, and the denoising function $\eta_t(\cdot)$ as listed in \cite[Section 1.1]{RushV16}, 
the sequence of the estimates $\{x^t\}$ that is generated by  AMP \eqref{eq:AMP1} \eqref{eq:AMP2}
has the following property~\cite{RushV16}. For all $\epsilon\in (0,1)$, there exist constants $K_t,\kappa_t>0$ independent of $n$ or $\epsilon$, such that
\begin{equation}
P\left(\left\lvert \frac{1}{N}\sum_{i=1}^{N}\phi(x^{t+1}_i,x_i) - \mathbb{E}\left[\phi(\eta_t(X+\tau^tZ),X)\right] \right\vert \geq \epsilon \right) \leq K_t e^{-\kappa_t n \epsilon^2},
\label{eq:thm_AMP_finte}
\end{equation}
where $\phi:\mathbb{R}^2\rightarrow\mathbb{R}$ is a pseudo-Lipschitz function of order 2 (PL(2)),\footnote{Recall the definition of PL(2) from \cite{Bayati2011}: a function $f:\mathbb{R}^m\rightarrow\mathbb{R}$ is said to be PL(2) if there is $L>0$ such that $|f(x)-f(y)|\leq L(1+\|x\|+\|y\|)\|x-y\|,\:\forall x,y\in\mathbb{R}^m$, where $\|\cdot\|$ denotes the Euclidean norm.} $X\sim p_X$, $Z$ is a standard normal random variable that is independent of $X$, and $\tau^t$ is defined via the following recursion ($(\sigma^0)^2=\delta^{-1}\mathbb{E}[X^2]$, $\delta=n/N$):
\begin{align}
(\tau^t)^2 &=\sigma_W^2 + (\sigma^t)^2,\nonumber\\
(\sigma^{t+1})^2 &=\delta^{-1}\mathbb{E}\left[\left(\eta_t(X+\tau^{t}Z)-X\right)^2\right].\label{eq:SP_SE}
\end{align} 

Notice that \eqref{eq:thm_AMP_finte} implies, by applying the Borel-Cantelli Lemma, the almost sure convergence result proved in \cite{Bayati2011}:
\begin{equation}
\lim_{N\rightarrow\infty}\frac{1}{N}\sum_{i=1}^{N}\phi(x^{t+1}_i,x_i)\overset{a.s.}{=}\mathbb{E}\left[\phi(\eta_t(X+\tau^tZ),X)\right].
\label{eq:thm_AMP}
\end{equation}

If we choose $\phi(x,y)=(x-y)^2$, then (\ref{eq:thm_AMP}) characterizes the mean square error (MSE)
achieved by AMP at each iteration.

\subsection{Column-Wise Multiprocessor AMP}
\label{subsec:MP_AMP}

In our proposed column-wise multiprocessor AMP (C-MP-AMP) algorithm,
the fusion center collects vectors that represent the estimations of the portion of the measurement vector $y$
contributed by the data from individual processors according to a pre-defined communication schedule.
The sum of these vectors is computed in the fusion center and transmitted to all processors.
Each processor performs standard AMP iterations with a new equivalent measurement vector, which
is computed using the vector received from the fusion center.
The pseudocode for C-MP-AMP is presented in Algorithm~1.

\begin{algorithm}
\caption{C-MP-AMP}
\textbf{Inputs to Processor $p$:} $y$, $A_p$, $\{\hat{k}_s\}_{s=0,...,\hat{s}}$ (maximum number of inner iterations at each outer iteration).\\
\textbf{Initialization:}  $x_p^{0,\hat{k}_0}=0$, $z_p^{0,\hat{k}_0-1}=0$, $r_p^{0,\hat{k}_0}=0$, $\forall p$.
\begin{algorithmic}
\NFor{$s=1:\hat{s}$} (loop over outer iterations)\\
$\quad$At the fusion center:
$\quad g^{s}=\sum_{u=1}^{P}r_u^{s-1,\hat{k}_{s-1}}$\\
$\quad$At Processor $p$:\\
$\quad x_p^{s,0}=x_p^{s-1,\hat{k}_{s-1}}$, $r_p^{s,0}=r_p^{s-1,\hat{k}_{s-1}}$
\NFor{$t=0:\hat{t}_s-1$} (loop over inner iterations)\\
$\quad\quad z_p^{s,k}=y-g_s-\left(r_p^{s,k}-r_p^{s,0}\right)$\\
$\quad\quad x_p^{s,k+1}=\eta_{s,k}(x_p^{s,k}+A_p^*z_p^{s,k})$\\
$\quad\quad r_p^{s,k+1}=A_px^{s,k+1}-\frac{z_p^{s,k}}{n}\sum_{i=1}^{N_p}\eta_{s,k}'([x_p^{s,k}+A_p^*z_p^{s,k}]_i)$.
\end{algorithmic}
\textbf{Output from Processor $p$:} $x_p^{\hat{s},\hat{k}_{\hat{s}}}$.
\label{algo:mp_amp_algo}
\end{algorithm}

\subsection{Performance Guarantee}
\label{subsec:SE}
Similar to AMP, the dynamics of the C-MP-AMP algorithm can be characterized by an SE formula. 
Let $(\sigma_p^{0,\hat{k}_0})^2=\delta_p^{-1}\mathbb{E}[X^2]$, where $\delta_p=n/N_p$, $\forall p=1,...,P$. 
For outer iterations $1\leq s\leq \hat{s}$ and inner iterations $0\leq t\leq\hat{k}_s$, we define the sequences $\{(\sigma_p^{s,k})^2\}$ and $\{(\tau_p^{s,k})^2\}$ as
\begin{align}
(\sigma_p^{s,0})^2&=(\sigma_p^{s-1,\hat{k}_s})^2,\label{eq:MP_SE1}\\
(\tau_p^{s,k})^2 &=\sigma_W^2+\sum_{u\neq p} (\sigma_u^{s,0})^2 + (\sigma_p^{s,k})^2,\label{eq:MP_SE2}\\
(\sigma_p^{s,k+1})^2 &=\delta_p^{-1}\mathbb{E}\left[\left(\eta_{s,k}(X+\tau_p^{s,k}Z)-X\right)^2 \right],\label{eq:MP_SE3}
\end{align}
where $Z$ is a standard normal random variable that is independent of $X$. 

With these definitions, we have the following performance guarantee for C-MP-AMP.
\begin{thm}
\label{thm:mainresult}
Under the assumptions listed in \cite[Section 1.1]{RushV16}, let $P$ be a fixed integer, for $p=1,...,P$, let $n/N_p=\delta_p\in (0,\infty)$ be a constant. Define $N=\sum_{p=1}^P N_p$. Then for any PL(2) function $\phi:\mathbb{R}^2\rightarrow\mathbb{R}$, we have $\forall \e \in (0,1)$, there exist constants $K_{s,k},\kappa_{s,k}>0$ independent of $n,\e$, such that
\begin{equation*}
P\left(\abs{\frac{1}{N_p}\sum_{i=1}^{N_p}\phi(x_{p,i}^{s,k+1},x_{p,i})-\mathbb{E}\left[\phi(\eta_{s,k}(X+\tau_p^{s,k}Z),X)\right]} \geq \e \right)\leq K_{s,k} e^{-\kappa_{s,k} n \e^2},\forall p,
\end{equation*}
where $x_p^{s,k+1}$ is generated by Algorithm \ref{algo:mp_amp_algo}, $\tau_p^{s,k}$ is defined in (\ref{eq:MP_SE1}--\ref{eq:MP_SE3}), $X\sim p_X$, and $Z$ is a standard normal random variable that is independent of $X$.
\end{thm} 

{\em Remark 1: C-MP-AMP converges to a fixed point that is no worse than that of AMP.} This statement can be demonstrated as follows.
When C-MP-AMP converges, the quantities in (\ref{eq:MP_SE1}--\ref{eq:MP_SE3}) do not keep changing, hence we can drop all the iteration indices for fixed point analysis. 
Notice that the last term on the right hand side (RHS) of (\ref{eq:MP_SE2}) vanishes, which leaves the RHS independent of $p$. 
That is, $(\tau_p^{s,k})^2$ are equal for all $p$, hence we can further drop the processor index for $(\tau_p^{s,k})^2$.
Denote $(\tau_p^{s,k})^2$ by $\tau^2$ for all $s,k,p$, and plug (\ref{eq:MP_SE3}) into (\ref{eq:MP_SE2}), then
\begin{align*}
\tau^2&=\sigma_W^2+\sum_{p=1}^{P}\delta_p^{-1}\mathbb{E}\left[\left(\eta(X+\tau Z)-X\right)^2\right]\\
&\overset{(a)}{=}\sigma_W^2+\delta^{-1}\mathbb{E}\left[\left(\eta(X+\tau Z)-X\right)^2\right],
\end{align*}
which is identical to the fixed point equation obtained from (\ref{eq:SP_SE}). In the above, step (a) holds because $\sum_{p=1}^P \delta_p^{-1}=\sum_{p=1}^P \frac{N_p}{n} = \frac{N}{n}.$
Because AMP always converges to the worst fixed point of the fixed point equation (\ref{eq:SP_SE})~\cite{Krzakala2012probabilistic}, 
the average asymptotic performance of C-MP-AMP is identical to AMP when there is only one solution to the
fixed point equation, and at least as good as AMP in case of multiple fixed points.

{\em Remark 2: The asymptotic dynamics of C-MP-AMP can be identical to AMP with a specific communication schedule.} This can be achieved by letting
$\hat{k}_s=1,\forall s$. In this case, the quantity $(\tau_p^{s,k})$ is involved only for $t=0$. 
Because the last term in (\ref{eq:MP_SE2}) is 0 when $t=0$, the computation
of $(\tau_p^{s,0})^2$ is independent of $p$. 
Therefore, $\tau_p^{s,0}$ are again equal for all $p$.
Dropping the processor index for $(\tau_p^{s,k})^2$, the recursion in (\ref{eq:MP_SE1}--\ref{eq:MP_SE3}) can be simplified as 
\begin{align*}
(\tau^{s,0})^2&=\sigma_W^2+\sum_{p=1}^P\delta_p^{-1}\mathbb{E}\left[\left(\eta_{s,0}(X+\tau^{s,0}Z)-X\right)^2\right]\\
&=\sigma_W^2+\delta^{-1}\mathbb{E}\left[\left(\eta_{s-1,0}(X+\tau^{s-1,0}Z)-X\right)^2\right],
\end{align*}
where the iteration evolves over $s$, which is identical to (\ref{eq:SP_SE}) evolving over $t$.

{\em Remark 3: Theorem \ref{thm:mainresult} implies almost sure convergence.} Similar to the performance guarantee for AMP \cite{RushV16}, the concentration result in Theorem \ref{thm:mainresult} implies
\begin{equation*}
\lim_{N\rightarrow\infty}\frac{1}{N_p}\sum_{i=1}^{N_p}\phi(x_{p,i}^{s,k+1},x_{p,i})\overset{a.s.}{=}\mathbb{E}\left[\phi(\eta_{s,k}(X+\tau_p^{s,k}Z),X)\right],\forall p,
\end{equation*}
by the Borel-Cantelli Lemma.

\section{Proofs of Theorem \ref{thm:mainresult}}
\label{sec:proof}

Our proof follows closely from the proof for AMP in \cite{RushV16}, with additional dependence structure to be addressed due to vectors being transmitted among processors. 

\subsection{Proof Notations}
\label{subsec:def}
Without loss of generality, we assume the sequence $\{\hat{k}_s\}_{s\geq 0}$ in Algorithm \ref{algo:mp_amp_algo} to be a constant value $\hat{k}$.
Let $t = s\hat{k} + k$, $\theta(t) = \lfloor t/\hat{k} \rfloor \hat{k}$. Given $w\in\mathbb{R}^n$, $x_p\in\mathbb{R}^{N_p}$, for $p=1,...,P$, define the column vectors $h_p^{t+1},q_p^{t}\in\mathbb{R}^{N_p}$ and $b_p^t,m_p^t\in\mathbb{R}^n$ for $t\geq 0$ recursively as follows. Starting with initial condition $q_p^0\in\mathbb{R}^{N_p}$: 
\begin{align}
&h_p^{t+1} = A_p^* m_p^t -  q_p^t,\qquad q_p^t = f_t(h_p^t,x_p)\nonumber\\
&b_p^t = A_p q_p^t-\lambda_p^t m_p^{t-1},\qquad m_p^t = b_p^t +\sum_{u\neq p} b_u^{\theta(t)}-w 
\label{eq:def_htbt}
\end{align}
where 
\begin{equation}
f_t(h^t,x_p) = \eta_{t-1}(x_p - h^t) - x_p,\qquad\text{and}\qquad\lambda_p^t:=\frac{1}{\delta_p N_p}\sum_{i=1}^{N_p} f_t'(h_{p,i}^t,x_{p,i}).
\label{eq:def_lambda_ft}
\end{equation}
In \eqref{eq:def_lambda_ft}, the derivative of $f_t:\mathbb{R}^2\rightarrow\mathbb{R}$ is with respect to the first argument. 
We assume that $\eta_t$ is Lipschitz for all $t\geq 0$, then it follows that $f_t$ is Lipschitz for all $t\geq 0$. Consequently, the weak derivative and $f_t'$ exit. Further, $f_t'$ is assumed to be differentiable, except possibly at a finite number of points, with bounded derivative whenever it exits. In \eqref{eq:def_htbt}, quantities with negative indices or with index $\theta(t)=0$ (i.e., $t<\hat{k}$) are defined to be zeros.

To see the equivalence between Algorithm \ref{algo:mp_amp_algo} and the recursion defined in \eqref{eq:def_htbt} and \eqref{eq:def_lambda_ft}, we let $x_p^0=0$, $r_p^0=0$, $z_p^t=0$, and 
\begin{align*}
h_p^{t+1} &= x_p - (A^* z_p^t + x_p^t),\qquad   q_p^t  = x_p^t - x_p,\\
b_p^t  &= r_p^t - A_p x_p,\qquad m_p^t  = -z_p^t.
\end{align*}

Let $(\sigma_p^0)^2 = \delta_p^{-1} \mathbb{E}[X^2]$. We assume that $(\sigma_p^0)^2$ is strictly positive for all $p=1,...,P$ and for all $\e \in (0,1)$, there exist $K,\kappa>0$ such that
\begin{equation}
\label{eq:q0_assume}
P\left(\abs{\frac{\|q_p^0\|^2}{n} - (\sigma_p^0)^2} \geq \e\right) \leq K e^{-\kappa n \e^2}, \quad\forall p = 1,...,P.
\end{equation}
Define the state evolution scalars $\{\tau_p^t\}_{t\geq 0}$ and $\{\sigma_p^t\}_{t\geq 1}$ for the the recursion defined in \eqref{eq:def_htbt} as follows:
\begin{equation}
\label{eq:se_general}
(\tau_p^t)^2 = (\sigma_p^t)^2 + \sum_{u\neq q} (\sigma_u^{\theta(t)})^2 + \sigma_W^2,\qquad (\sigma_p^t)^2 = \frac{1}{\delta_p}\mathbb{E}\left[\left(f_t(\tau_p^{t-1}Z,X)\right)^2  \right],
\end{equation}
where $Z\sim\mc{N}(0,1)$ and $X\sim p_X$ are independent. Notice that with the equivalence between Algorithm \ref{algo:mp_amp_algo} and the recursion \ref{eq:def_htbt}, the state evolution scalars defined in \eqref{eq:se_general} matches \eqref{eq:MP_SE1} - \eqref{eq:MP_SE3}.

Writing the updating equations for $b_p^t, h_p^{t+1}$ defined in \eqref{eq:def_htbt} in matrix form, we have
\begin{equation}
X_p^t = A_p^* M_p^t,\quad Y_p^t = A_p Q_p^t, \label{eq:update_mtx_form}
\end{equation}
where
\begin{align*}
X_p^t &= [ h_p^1 +  q_p^0 | h_p^2 +  q_p^1 |\cdots| h_p^t +  q_p^{t-1} ],\qquad 
Y_p^t = [ b_p^0 | b_p^1 + \lambda_p^1 m_p^{0} | \cdots | b_p^{t-1} + \lambda_p^{t-1} m_p^{t-2}]\\
M_p^t &=[ m_p^0 | m_p^1 | \cdots |m_p^{t-1}],\qquad \qquad
Q_p^t=[ q_p^0 | q_p^1 | \cdots| q_p^{t-1}].
\end{align*}

Let $(m_p^t)_{||}$ and $(q_p^t)_{||}$ denote the projection of $m_p^t$ and $q_p^t$ onto the column space of $M_p^t$ and $Q_p^t$, respectively.
That is,
\begin{align*}
(m_p^t)_{||}&=M_p^t\left((M_p^t)^*M_p^t\right)^{-1}(M_p^t)^*m_p^t\\
(q_p^t)_{||}&=Q_p^t\left((Q_p^t)^*Q_p^t\right)^{-1}(Q_p^t)^*q_p^t.
\end{align*}
Let
\begin{equation}
\alpha_p^t=(\alpha_{p,0}^t,\alpha_{p,1}^t,...,\alpha_{p,t-1}^t)^*,\quad \gamma_p^t=(\gamma_{p,0}^t,\gamma_{p,1}^t,...,\gamma_{p,t-1}^t)^*
\end{equation}
be the coefficient vectors of these projections. That is,
\begin{equation}
\alpha_p^t = \left((M_p^t)^*M_p^t\right)^{-1}(M_p^t)^*m_p^t,\quad
\gamma_p^t = \left((Q_p^t)^*Q_p^t\right)^{-1}(Q_p^t)^*q_p^t.
\end{equation}
and
\begin{equation}
(m_p^t)_{||}=\sum_{i=0}^{t-1}\alpha_{p,i}^tm_p^i,\quad
(q_p^t)_{||}=\sum_{i=0}^{t-1}\gamma_{p,i}^tq_p^i.
\end{equation}
Define
\begin{equation}
(m_p^t)_{\perp}=m_p^t-(m_p^t)_{||},\quad
(q_p^t)_{\perp}=q_p^t-(q_p^t)_{||}.
\end{equation}

The main lemma will show that $\alpha_p^t$ and $\gamma_p^t$ concentrate around some constant $\hat{\alpha}_p^t$ and $\hat{\gamma}_p^t$, respectively. We define these constants in the following subsection. 

\subsection{Concentrating Constants}
\label{subsec:conc_const}

Let $\{\tilde{Z}_p^t\}_{t\geq 0}$ and $\{\breve{Z}_p^t\}_{t\geq 0}$ each be a sequence of zero-mean jointly Gaussian random variables whose covariance is defined recursively as follows. For $t,r\geq 0$,
\begin{equation}
\mathbb{E}[\breve{Z}_p^r\breve{Z}_p^t] =\frac{\tilde{E}_p^{r,t}}{\sigma_p^r \sigma_p^t},\qquad \mathbb{E}[\tilde{Z}_p^r\tilde{Z}_p^t] =\frac{\breve{E}_p^{r,t}}{\tau_p^r\tau_p^t},
\label{eq:def_Z_tilde}
\end{equation}
where
\begin{align}
\breve{E}_p^{r,t}&= \tilde{E}_p^{r,t} + \sum_{u\neq p} \tilde{E}_u^{\theta(r),\theta(t)} + \sigma_W^2,\nonumber\\
\tilde{E}_p^{r,t}&=\delta_p^{-1}\mathbb{E}\left[f_r(\tau_p^{r-1}\tilde{Z}_p^{r-1},X)f_t(\tau_p^{t-1}\tilde{Z}_p^{t-1},X)\right].
\label{eq:def_tilde_E}
\end{align}
Moreover, $\tilde{Z}_p^r$ is independent of $\tilde{Z}_q^t$ and $\breve{Z}_p^r$ is independent of $\breve{Z}_q^t$ for all $r,t\geq 0$ whenever $p\neq q$.
Note that according to the definition of $\sigma_p^t$ and $\tau_p^t$ in (\ref{eq:se_general}), we have $\breve{E}_p^{t,t}=(\tau_p^t)^2$, $\tilde{E}_p^{t,t}=(\sigma_p^t)^2$, and $\mathbb{E}[(\tilde{Z}_p^t)^2]=\mathbb{E}[(\breve{Z}_p^t)^2]=1$. In \eqref{eq:def_tilde_E}, quantities with negative indices or with either $\theta(t)=0$ or $\theta(r)=0$ are zeros.

Define matrices $\tilde{C}_p^t,\breve{C}_p^t\in\mathbb{R}^{t\times t}, p=1,2$, such that
\begin{equation*}
[\tilde{C}_p^t]_{r+1,s+1} =\tilde{E}_p^{r,s}, \quad
[\breve{C}_p^t]_{r+1,s+1} = \breve{E}_p^{r,s}, \forall r,s=0,...,t-1.
\end{equation*}

Define vectors $\tilde{E}_p^t,\breve{E}_p^t\in\mathbb{R}^t, p=1,2$, such that
\begin{equation*}
\tilde{E}_p^t=(\tilde{E}_p^{0,t},\tilde{E}_p^{1,t},...,\tilde{E}_p^{t-1,t}),\quad
\breve{E}_p^t=(\breve{E}_p^{0,t},\breve{E}_p^{1,t},...,\breve{E}_p^{t-1,t}).
\end{equation*}

Define the concentrating values $\hat{\alpha}_p^t$ and $\hat{\gamma}_p^t$ as
\begin{equation}
\hat{\gamma}_p^t=(\tilde{C}_p^t)^{-1}\tilde{E}_p^t,\quad \hat{\alpha}_p^t=(\breve{C}_p^t)^{-1}\breve{E}_p^t.
\label{eq:def_hat_gamma_alpha}
\end{equation}

Let $(\sigma_p^0)_{\perp}^2=(\sigma_p^0)^2$ and $(\tau_p^0)_{\perp}^2=(\tau_p^0)^2$, and for $t>0$, define
\begin{align}
(\sigma_p^t)_{\perp}^2 &= (\sigma_p^t)^2-(\hat{\gamma}_p^t)^*\tilde{E}_p^t=(\sigma_p^t)^2-(\tilde{E}_p^t)^*(\tilde{C}_p^t)^{-1}\tilde{E}_p^t,\nonumber\\
(\tau_p^t)_{\perp}^2 &= (\tau_p^t)^2-(\hat{\alpha}_p^t)^*\breve{E}_p^t=(\sigma_p^t)^2-(\breve{E}_p^t)^*(\breve{C}_p^t)^{-1}\breve{E}_p^t.
\label{eq:def_sigma_tau_perp}
\end{align}

\begin{lem}
\label{lemma:invertible}
The matrices $\tilde{C}_p^t$ and $\breve{C}_p^t$, $\forall t\geq 0$, defined above are invertible, and the scalars $(\sigma_p^t)^2_{\perp}$ and $(\tau_p^t)_{\perp}^2$, $\forall t\geq 0$, defined above are strictly positive. 
\end{lem}
\begin{proof}
The proof for $\tilde{C}_p^t$ being invertible and $(\sigma_p^t)^2_{\perp}$ being strictly positive is the same as in \cite{RushV16}. 
Now consider $\breve{C}_p^{t+1}$. Notice that $\breve{C}_p^{t+1}$ is the sum of a positive definite matrix ($\tilde{C}_p^{t+1}$) and $P$ positive semi-definite matrices, hence,
$\breve{C}_p^{t+1}$ is positive definite.
Consequently, 
\begin{equation}
\det(\breve{C}_p^{t+1})=\det(\breve{C}_p^t)\det((\tau_p^t)^2-(\breve{E}_p^t)^*(\breve{C}_p^t)^{-1}\breve{E}_p^t)>0,
\end{equation}
which implies $(\tau_p^t)^2-(\breve{E}_p^t)^*(\breve{C}_p^t)^{-1}\breve{E}_p^t=(\tau_p^t)^2_\perp>0$.
\end{proof}

\subsection{Condition Distribution Lemma}
\label{subsec:cond_lemma}
Let the sigma algebra $\mathscr{S}^{t_1,t}$ be generated by $x,w,b_p^0,...,b_p^{t_1-1},m_p^{0},...,m_p^{t_1-1},h_p^1,...,h_p^{t},q_p^0,...,q_p^{t}$, $\forall p$.
We now compute the conditional distribution of $A_p$ given $\mathscr{S}^{t_1,t}$ for $1 \leq p \leq P$, where $t_1$ is either $t$ or $t+1$.

Notice that conditioning on $\mathscr{S}^{t_1,t}$ is equivalent to conditioning on the linear constraints:
\begin{equation}
A_p Q_p^{t_1} = Y_p^{t_1},\quad A_p^* M_p^t = X_p^t,\qquad 1 \leq p \leq P,
\label{eq:linearConstraints}
\end{equation}
where in \eqref{eq:linearConstraints}, only $A_p$, $1 \leq p \leq P$, are treated as random.

Let $\mathsf{P}_{Q_p^{t_1}}^{\parallel}=Q_p^{t_1}((Q_p^{t_1})^*Q_p^{t_1})^{-1}Q_p^{t_1}$ and $\mathsf{P}_{M_p^t}^{\parallel}=M_p^t((M_p^t)^*M_p^t)^{-1}M_p^t$, which are the projectors onto the column space of $Q_p^{t_1}$ and $M_p^t$, respectively. The following lemma provides the conditional distribution of the matrices $A_p$, $p=1,...,P$, given $\mc{G}^{t_1,t}$.
\begin{lem}
\label{lem:Ap_cond}
For $t_1=t \text{ or } t+1$, the conditional distribution of the random matrices $A_p$, $p=1,...,P$, given $\mathscr{S}^{t_1,t}$ satisfies
\begin{equation*}
(A_1,...,A_P)|_{\mathscr{S}^{t_1,t}} \overset{d}{=} (\mathsf{E}_1^{t_1,t} + \mathsf{P}_{M_1^t}^{\perp} \tilde{A}_1 \mathsf{P}_{Q_1^{t_1}}^{\perp},...,\mathsf{E}_P^{t_1,t} + \mathsf{P}_{M_P^t}^{\perp} \tilde{A}_P \mathsf{P}_{Q_P^{t_1}}^{\perp})
\end{equation*}
where $\mathsf{P}_{Q_p^{t_1}}^{\perp} = \mathsf{I}-\mathsf{P}_{Q_p^{t_1}}^{\parallel}$ and $\mathsf{P}_{M_p^t}^{\perp} = \mathsf{I}-\mathsf{P}_{M_p^t}^{\parallel}$. $\tilde{A}_p\overset{d}{=}A_p$ and $\tilde{A}_p$ is independent of $\mathscr{S}^{t_1,t}$. Moreover, $\tilde{A}_p$ is independent of $\tilde{A}_q$ for $p\neq q$. $E_p^{t_1,t}$ is defined as
\begin{align*}
\mathsf{E}_p^{t_1,t} &= Y_p^{t_1}((Q_p^{t_1})^*Q_p^{t_1})^{-1}(Q_p^{t_1})^* +  M_p^t((M_p^t)^*M_p^t)^{-1}(X_p^t)^*\\
&\hspace{1in} - M_p^t((M_p^t)^*M_p^t)^{-1}(M_p^t)^*Y_p^{t_1}((Q_p^{t_1})^*Q_p^{t_1})^{-1}(Q_p^{t_1})^*.
\end{align*}
\end{lem}

\begin{proof}
To simplify the notation, we drop the superscript $t$ or $t_1$ in the following proof. It should be understood that $Q_p$ represents $Q_p^{t_1}$, $Y_p$ represents $Y_p^{t_1}$, $M_p$ represents $M_p^t$, and $X_p$ represents $X_p^t$.

First let us consider projections of a deterministic matrix.
Let $\hat{A}_p$ be a deterministic matrix that satisfies the linear constraints $Y_p = \hat{A}_p Q_p$ and $X_p = \hat{A}_p^* M_p$, then we have
\begin{align*}
\hat{A}_p &= \hat{A}_pQ_p(Q_p^*Q_p)^{-1}Q_p^* + \hat{A}_p(\mathsf{I} - Q_p(Q_p^*Q_p)^{-1}Q_p^*),\\
\hat{A}_p &= M_p(M_p^*M_p)^{-1}M_p^*\hat{A}_p + (\mathsf{I} - M_p(M_p^*M_p)^{-1}M_p^*)\hat{A}_p.
\end{align*}
Combining the two equations above, as well as the two linear constraints, we can write
\begin{equation}
\hat{A}_p = Y_p(Q_p^*Q_p)^{-1}Q_p^* + M_p(M_p^*M_p)^{-1}X_p - M_p(M_p^*M_p)^{-1}M_p^*Y_p(Q_p^*Q_p)^{-1}Q_p^* + \mathsf{P}_{M_p}^{\perp} \hat{A}_p \mathsf{P}_{Q_p}^{\perp}.
\label{eq:projA}
\end{equation}

We now demonstrate the conditional distribution of $A_1,...,A_P$. Let $S_1,...,S_P$ be arbitrary Borel sets on $\mathbb{R}^{n\times N_1}$,...,$\mathbb{R}^{n\times N_P}$, respectively.
\begin{align}
&P\left( A_1\in S_1,..., A_P\in S_P \left\vert A_p Q_p = Y_p, A_p^* M_p = X_p, \forall p \right.\right)\nonumber\\
&\overset{(a)}{=} P\left( \mathsf{E}_1^{t_1,t} + \mathsf{P}_{M_1}^{\perp} A_1 \mathsf{P}_{Q_1}^{\perp}\in S_1,...,\mathsf{E}_P^{t_1,t} + \mathsf{P}_{M_P}^{\perp} A_P \mathsf{P}_{Q_P}^{\perp}\in S_P  \left\vert A_p Q_p = Y_p, A_p^* M_p = X_p, \forall p \right.\right)\nonumber\\
&\overset{(b)}{=} P\left( \mathsf{E}_1^{t_1,t} + \mathsf{P}_{M_1}^{\perp} A_1 \mathsf{P}_{Q_1}^{\perp}\in S_1,...,\mathsf{E}_P^{t_1,t} + \mathsf{P}_{M_P}^{\perp} A_P \mathsf{P}_{Q_P}^{\perp}\in S_P \right)\nonumber\\
& = P\left( \mathsf{E}_1^{t_1,t} + \mathsf{P}_{M_1}^{\perp} A_1 \mathsf{P}_{Q_1}^{\perp}\in S_1\right)...P\left(\mathsf{E}_P^{t_1,t} + \mathsf{P}_{M_P}^{\perp} A_P \mathsf{P}_{Q_P}^{\perp}\in S_P \right),\label{eq:cond_A_p}
\end{align}
which implies the desired result.
In step (a), 
\begin{equation*}
\mathsf{E}_p^{t_1,t} = Y_p(Q_p^*Q_p)^{-1}Q_p^* +  M_p(M_p^*M_p)^{-1}X_p^* - M_p(M_p^*M_p)^{-1}M_p^*Y_p (Q_p^*Q_p)^{-1}Q_p^*,\quad p=1,...,P,
\end{equation*}
which follows from \eqref{eq:projA}.
Step (b) holds since $\mathsf{P}_{M_p}^{\perp} A_p \mathsf{P}_{Q_p}^{\perp}$ is independent of the conditioning. The independence is demonstrated as follows.
Notice that $A_p Q_p = A_p \mathsf{P}_{Q_p}^{||} Q_p$. In what follows, we will show that
$A_p^{||}:=A_p \mathsf{P}_{Q_p}^{||}$ is independent of $A_r^{\perp}:=A_r \mathsf{P}_{Q_r}^{\perp}$, for $p,r = 1,...,P$. Then similar approach can be used to demonstrate that $\mathsf{P}_{M_p}^{\perp}A_p$ is independent of $\mathsf{P}_{M_r}^{||}A_r$. Together they provide the justification for step (b). Note that $A_p^{||}$ and $A_r^{\perp}$ are jointly normal, hence it is enough to show they are uncorrelated.
\begin{align*}
\mathbb{E}\left\{ [A_p^{||}]_{i,j} [A_r^{\perp}]_{m,l} \right\} &= \mathbb{E}\left\{ \left(\sum_{k=1}^N [A_p]_{i,k} [\mathsf{P}_{Q_p}^{\parallel}]_{k,j} \right) \left( \sum_{k=1}^N [A_r]_{m,k} \left( \mathsf{I}_{k,l} - [\mathsf{P}_{Q_r}^{\parallel}]_{k,l} \right) \right) \right\}\\
&\overset{(a)}{=} \frac{1}{n}\delta_0(i,m)\delta_0(p,r) \left(\sum_{k=1}^N  [\mathsf{P}_{Q_p}^{\parallel}]_{k,j}\mathsf{I}_{k,l} -  \sum_{k=1}^N  [\mathsf{P}_{Q_p}^{\parallel}]_{k,j}[\mathsf{P}_{Q_r^{t_1}}^{\parallel}]_{k,l} \right)\\
&\overset{(b)}{=} \frac{1}{n}\delta_0(i,m)\delta_0(p,r)  \left( [\mathsf{P}_{Q_p}^{\parallel}]_{l,j} -  \sum_{k=1}^N  [\mathsf{P}_{Q_p}^{\parallel}]_{k,j}[\mathsf{P}_{Q_r}^{\parallel}]_{l,k} \right)\overset{(c)}{=} 0,
\end{align*}
where $\delta_0(i,j)$ is the Kronecker delta function.
In the above, step (a) holds since the original matrix $A$ has $\mathcal{N}(0,1/n)$ i.i.d. entries, step (b) holds since projectors are symmetric matrices, and step (c) follows
the property of projectors $\mathsf{P}^2 = \mathsf{P}$.
\end{proof}

Combining the results in Lemma \ref{lem:Ap_cond} and \cite[Lemma 4]{RushV16}, we have the following conditional distribution lemma.
\begin{lem}
\label{lem:cond_htbt}
For the vectors $h_p^{t+1}$ and $b_p^t$ defined in \eqref{eq:def_htbt}, the following holds for $t\geq 1$, $p=1,...,P$:
\begin{align}
b_p^0|\mathscr{S}^{0,0}& \overset{d}{=} (\sigma_p^0)_{\perp}Z_p^{'0}+\Delta_p^{0,0},\qquad\qquad
h_p^1|\mathscr{S}^{1,0}\overset{d}{=}  (\tau_p^0)_{\perp}Z_p^{0}+\Delta_p^{1,0},\label{eq:def_cond_b0}\\
b_p^{t}|\mathscr{S}^{t,t}&\overset{d}{=}  \sum_{i=0}^{t-1}\hat{\gamma}_{p,i}^tb_p^i + (\sigma_p^t)_{\perp}Z_p^{'t}+\Delta_p^{t,t},\qquad
h_p^{t}|\mathscr{S}^{t+1,t}\overset{d}{=}  \sum_{i=0}^{t-1}\hat{\alpha}_{p,i}^th_p^{i+1} + (\tau_p^t)_{\perp}Z_p^{t}+\Delta_p^{t+1,t},\label{eq:def_cond_bt}
\end{align}
where
\begin{align}
\Delta_p^{0,0}&=\left(\frac{\|(q_p^0)_{\perp}\|}{\sqrt{n}}-(\sigma_p^0)_{\perp}\right)Z_p^{'0} \label{eq:def_del00}\\
\Delta_p^{1,0}&=\left[\left(\frac{\|(m_p^0)_{\perp}\|}{\sqrt{n}}-(\tau_p^0)_{\perp}\right)\mathsf{I}-\frac{\|(m_p^0)_{\perp}\|}{\sqrt{n}}\mathsf{P}_{q_p^0}^\parallel\right]Z_p^{0}
+q_p^0\left(\frac{\|q_p^0\|^2}{n}\right)^{-1}\left(\frac{(b_p^0)^*(m_p^0)_{\perp}}{n}-\frac{\|q_p^0\|^2}{n}\right)\\
\Delta_p^{t,t}&=\sum_{i=0}^{t-1}\left(\gamma_{p,i}^t-\hat{\gamma}_{p,i}^t\right)b_p^i
+\left[\left(\frac{\|(q_p^t)_{\perp}\|}{\sqrt{n}}-(\sigma_p^t)_{\perp}\right)\mathsf{I}-\frac{\|(q_p^t)_\perp\|}{\sqrt{n}}\mathsf{P}_{M_p^t}^\parallel\right]Z_p^{'t}\nonumber\\
&+M_p^t\left(\frac{(M_p^t)^*M_p^t}{n}\right)^{-1}
\left(\frac{({H}_p^t)^*(q_p^t)_{\perp}}{n}-\frac{(M_p^t)^*}{n}\left[\lambda_p^tm_p^{t-1}-\sum_{i=1}^{t-1}\lambda_{p,i}^t\gamma_{p,i}^tm_p^{i-1} \right]\right)\label{eq:def_deltt}\\
\Delta_p^{t+1,t}&=\sum_{i=0}^{t-1}\left(\alpha_{p,i}^t-\hat{\alpha}_{p,i}^t\right)h_p^{i+1}
+\left[\left(\frac{\|(m_p^t)_{\perp}\|}{\sqrt{n}}-(\tau_p^t)_{\perp}\right)\mathsf{I}-\frac{\|(m_p^t)_{\perp}\|}{\sqrt{n}}\mathsf{P}_{Q_p^{t+1}}^\parallel\right]Z_p^{t}\nonumber\\
&+Q_p^{t+1}\left(\frac{(Q_p^{t+1})^*Q_p^{t+1}}{n}\right)^{-1}
\left(\frac{({ B}_p^{t+1})^*(m_p^t)_{\perp}}{n}-\frac{(Q_p^{t+1})^*}{n}\left[q_p^{t}-\sum_{i=0}^{t-1}\alpha_{p,i}^tq_p^{i} \right]\right),
\end{align}
where $Z_p^{'t}\in\mathbb{R}^n$ and $Z_p^t\in\mathbb{R}^{N_p}$ are random vectors with independent standard normal elements, and are independent of the corresponding sigma algebras.
Moreover, $Z_p^{'t}$ is independent of $Z_q^{'t}$ and $Z_p^{t}$ is independent of $Z_q^{t}$ when $p\neq q$.
\end{lem}
\begin{proof}
The proof for each individual $p\in [P]$ is similar to the proof for \cite[Lemma 4]{RushV16}. The claim that $Z_p^{'t}$ is independent of $Z_q^{'t}$ and $Z_p^{t}$ is independent of $Z_q^{t}$ when $p\neq q$ follows from Lemma \ref{lem:Ap_cond}, where we have that $\tilde{A}_p$ is independent of $\tilde{A}_q$ for $p\neq q$.
\end{proof}

\subsection{Main Concentration Lemma}
\label{subsec:main_lemma}

We use the shorthand $X_n \doteq c$ to denote the concentration inequality $P\left(\abs{X_n-c}\geq \e\right)\leq K_t e^{-\kappa_t n \e}$. As specified in the theorem statement, the lemma holds for all $\e \in (0,1)$, with $K_t,\kappa_t$ denoting the generic constants dependent on $t$, but not on $n,\e$.

\begin{lem}
With the $\doteq$ notation defined above, the following holds for all $t\geq 0$, $p=1,...,P$.

\begin{enumerate}
\item[(a)] 
\begin{align}
P\left(\frac{\|\Delta_p^{t,t}\|^2}{n}\geq\epsilon\right)&\leq K_t e^{-\kappa_t n\epsilon}.\\
P\left(\frac{\|\Delta_p^{t+1,t}\|^2}{n}\geq\epsilon\right)&\leq K_t e^{-\kappa_t n\epsilon}.\label{eq:main_lem_Ba}
\end{align}
\item[(b)]
({\em i}) For pseudo Lipschitz functions $\phi_h:\mathbb{R}^{t+2}\rightarrow\mathbb{R}$,
\begin{equation}
\frac{1}{N_p}\sum_{i=1}^{N_p}\phi_h(h_{p,i}^1,...,h_{p,i}^{t+1},x_{p,i}) \doteq \mathbb{E}[\phi_h(\tau_p^0\tilde{Z}_p^0,...,\tau_p^t\tilde{Z}_p^t,X)].
\end{equation}
({\em ii}) Let $\psi_h:\mathbb{R}^2\rightarrow\mathbb{R}$ be a bounded function that is differentiable in the first argument except possibly at a finite number of points, with bounded derivative when it exists. Then,
\begin{equation}
\frac{1}{N_p}\sum_{i=1}^{N_p}\psi_h(h_{p,i}^{t+1},x_{p,i}) \doteq \mathbb{E}[\psi_h(\tau_{p}^t\tilde{Z}_p^t,X)],
\end{equation}
where $\{\tilde{Z}_p^t\}$ is defined in \eqref{eq:def_Z_tilde}, and $X\sim p_X$ is independent of $\{\tilde{Z}\}_p^t$.\\
({\em iii}) For pseudo-Lipschitz function $phi_b:\mathbb{R}^{P(t+1)+1}\rightarrow\mathbb{R}$,
\begin{equation}
\frac{1}{n}\sum_{i=1}^n \phi_b(b_{1,i}^0,...,b_{P,i}^0,...,b_{1,i}^t,...,b_{P,i}^t,w_i) \doteq \mathbb{E}\left[\phi_b(\sigma_1^0\breve{Z}_1^0,...,\sigma_P^0\breve{Z}_P^0,...,\sigma_1^t\breve{Z}_1^t,...,\sigma_P^t\breve{Z}_P^t,W)\right],\label{eq:main_lem_Bb3}
\end{equation}
where $\{\breve{Z}_p^t\}$ is defined in \eqref{eq:def_Z_tilde}, and $W\sim p_W$ is independent of $\{\breve{Z}_p^t\}$.
\item[(c)]
\begin{align}
\frac{(h_p^{t+1})^*q_p^0}{n} &\doteq 0,\qquad \frac{(h_p^{t+1})^*x_p}{n}\doteq 0.\\
\frac{(b_p^t)^* w}{n} &\doteq 0. \label{eq:main_lem_Bc}
\end{align}
\item[(d)] For all $0\leq r\leq t$, $q\neq p$,
\begin{align}
\frac{(h_p^{r+1})^*h_p^{t+1}}{n} &\doteq \breve{E}_p^{r,t}.\\
\frac{(b_p^r)^*b_p^t}{n} &\doteq \tilde{E}_p^{r,t}.\label{eq:main_lem_Bd}
\end{align}
\item[(e)] For all $0\leq r\leq t$, 
\begin{align}
\frac{(q_p^{0})^*q_p^{t+1}}{n} &\doteq \tilde{E}_p^{r+1,t+1},\qquad \frac{(q_p^{r+1})^*q_p^{t+1}}{n} \doteq \tilde{E}_p^{r+1,t+1}.\\
\frac{(m_p^r)^*m_p^t}{n} &\doteq \breve{E}_p^{r,t}. \label{eq:main_lem_Be}
\end{align}
\item[(f)] Define $\hat{\lambda}_p^{t+1}=\delta_p^{-1}\mathbb{E}[f_t'(\tau_p^t\tilde{Z}_p^t,X)]$.  For all $ 0\leq r\leq t$, 
\begin{align}
\lambda_p^t \doteq \hat{\lambda}_p^t, \qquad \frac{(h_p^{r+1})^*q_p^{t+1}}{n} &\doteq \hat{\lambda}_p^{t+1}\breve{E}_p^{r,t},\qquad \frac{(h_p^{t+1})^*q_p^{r+1}}{n} \doteq \hat{\lambda}_p^{r+1}\breve{E}_p^{r,t}.\\
\frac{(b_p^{r})^*m_p^{t}}{n} &\doteq \tilde{E}_p^{r,t},\qquad \frac{(b_p^{t})^*m_p^{r}}{n} \doteq \tilde{E}_p^{r,t}.\label{eq:main_lem_Bf}\\
\end{align}
\item[(g)] For $\mathbf{Q}_p^{t+1} = \frac{1}{n}(Q_p^{t+1})^*Q_p^{t+1}$ and $\mathbf{M}_p^{t} = \frac{1}{n}(M_p^{t})^*M_p^{t}$, when the inverses exist, for all $0\leq i,j\leq t$, $0\leq i',j'\leq t-1$, 
\begin{align}
[(\mathbf{Q}_p^{t+1})^{-1}]_{i+1,j+1} &\doteq [(\tilde{C}_p^{t+1})^{-1}]_{i+1,j+1}, \qquad  \gamma_{p,i}^{t+1} \doteq \hat{\gamma}_{p,i}^{t+1}.\\
[(\mathbf{M}_p^{t})^{-1}]_{i'+1,j'+1} &\doteq [(\breve{C}_p^{t})^{-1}]_{i'+1,j'+1}, \qquad   \alpha_{p,i'}^t \doteq \hat{\alpha}_{p,i'}^t, \label{eq:main_lem_Bg}
\end{align}
where $\hat{\gamma}_{p,i}^{t+1}$ and $\hat{\alpha}_{p,i}^t$ are defined in \eqref{eq:def_hat_gamma_alpha}.
\item[(h)]
\begin{align}
\frac{\|(q_p^{t+1})_{\perp}\|^2}{n} &\doteq (\sigma_p^{t+1})_\perp^2.\\
\frac{\|(m_p^t)_{\perp}\|^2}{n} &\doteq (\tau_p^t)_\perp^2, \label{eq:main_lem_Bh}
\end{align}
where $(\sigma_p^{t+1})_\perp^2$ and $(\tau_p^t)_\perp^2$ are defined in \eqref{eq:def_sigma_tau_perp}.
\end{enumerate}
\label{lem:main_lemma}
\end{lem}

\subsection{Proof of Theorem \ref{thm:mainresult}}
\label{subsec:proof_thm1}
Applying Part (b)(i) of Lemma \ref{lem:main_lemma}, the proof the Theorem \ref{thm:mainresult} is the same as the proof in \cite{RushV16} and therefore is not repeated here.

\section{Proof of Lemma \ref{lem:main_lemma}}
\label{sec:proof_mainlem}

The proof of Lemma \ref{lem:main_lemma} uses induction in the iteration index $t$. We label as $\mc{B}_t$ the results \eqref{eq:main_lem_Ba}, \eqref{eq:main_lem_Bb3}, \eqref{eq:main_lem_Bc}, \eqref{eq:main_lem_Bd}, \eqref{eq:main_lem_Be}, \eqref{eq:main_lem_Bf}, \eqref{eq:main_lem_Bg}, \eqref{eq:main_lem_Bh}. The rest of the results in Lemma \ref{lem:main_lemma} are labels as $\mc{H}_{t+1}$. In the proof, $K,\kappa$ are used as universal constant in the upper bounds.

\subsection{Step 1: Showing $\mc{B}_0$ holds.} In the following, when a statement involves $p=1,...,P$, we demonstrate the proof for $p=1$, and the cases when $p=2,...,P$ can be obtained similarly; unless otherwise specified.
 
\textbf{(a)} Recall the definition of $\Delta_1^{0,0}$ in \eqref{eq:def_del00}, we have
\begin{align}
P\left( \frac{\|\Delta_1^{0,0}\|}{n} \geq \e \right) &\leq P\left( \abs{\frac{\|q_1^0\|}{\sqrt{n}} - (\sigma_1^0)_{\perp}} \geq \sqrt{\frac{\e}{2}} \right) + P\left( \abs{\frac{\|Z_1^{'0}\|}{\sqrt{n}} - 1} \geq \sqrt{\frac{\e}{2}} \right)\nonumber\\ 
&\overset{(a)}{\leq} Ke^{-\kappa n \e} + Ke^{-\kappa n \e},
\label{eq:B0_a1}
\end{align}
where step $(a)$ uses our assumption on $q_1^0$ in \eqref{eq:q0_assume}, Lemma \ref{sqroots}, and Lemma \ref{subexp}.

\textbf{(b) (iii)} Recall the conditional distribution of $b_p^0$ in \eqref{eq:def_cond_b0}, by Lemma \ref{sums},
\begin{align*}
&P\left(\abs{\frac{1}{n}\sum_{i=1}^n \phi_b(b_{1,i}^0,...,b_{P,i}^0,w_i) - \mathbb{E}\left[\phi_b(\sigma_1^0\breve{Z}_1^0,...,\sigma_P^0\breve{Z}_P^0,W)\right]} \geq \e \right)\\
&\leq P\left(\frac{1}{n}\sum_{i=1}^n \abs{\phi_b(\sigma_1^0 Z_{1,i}^{'0}+\Delta_{1,i}^{0,0},...,\sigma_P^0 Z_{P,i}^{'0}+\Delta_{P,i}^{0,0},w_i) - \phi_b(\sigma_1^0 Z_{1,i}^{'0},...,\sigma_P^0 Z_{P,i}^{'0},w_i)} \geq \frac{\e}{2} \right)\\
&+ P\left(\abs{\frac{1}{n}\sum_{i=1}^n \phi_b(\sigma_1^0 Z_{1,i}^{'0},...,\sigma_P^0 Z_{P,i}^{'0},w_i) - \mathbb{E}\left[\phi_b(\sigma_1^0\breve{Z}_1^0,...,\sigma_P^0\breve{Z}_P^0,W)\right]} \geq \frac{\e}{2} \right).
\end{align*}
Label the two terms above as $T_1$ and $T_2$, we will show that both are upper bounded by $Ke^{-\kappa n \e^2}$.
First consider $T_1$.
\begin{align}
T_1 &\overset{(a)}{\leq} P\left( \frac{1}{n} \sum_{i=1}^n \left(1+ \sqrt{\sum_{u=1}^P \left(\sigma_u^0 Z_{u,i}^{'0} + \Delta_{u,i}^{0,0}\right)^2} +  \sqrt{\sum_{u=1}^P \left(\sigma_u^0 Z_{u,i}^{'0}\right)^2} \right)\left(\sqrt{\sum_{u=1}^P (\Delta_{u,i}^{0,0})^2}\right) \geq \frac{\e}{2L}\right)\nonumber\\
&\overset{(b)}{\leq} P\left(\sqrt{1+ \sum_{u=1}^P \frac{\|\Delta_u^{0,0}\|^2}{n} + 4 \sum_{u=1}^{P} (\sigma_u^0)^2 \frac{\|Z_u^{'0}\|^2}{n}} \sqrt{\sum_{u=1}^P \frac{\|\Delta_u^{0,0}\|^2}{n}}  \geq \frac{\e}{2\sqrt{3}L}\right)\nonumber\\
&\leq P\left( \left( 1 + \sum_{u=1}^P \frac{\|\Delta_u^{0,0}\|}{\sqrt{n}} + 2\sum_{u=1}^P \sigma_u^0\frac{\|Z_u^{'0}\|}{\sqrt{n}}\right) \left(\sum_{u=1}^P \frac{\|\Delta_u^{0,0}\|}{\sqrt{n}}\right) \geq \frac{\e}{2\sqrt{3}L} \right).\label{eq:B0_b1}
\end{align}
In the above, step $(a)$ holds because $\phi_b$ is pseudo-Lipschitz. Step $(b)$ uses triangle inequality to obtain $\sqrt{\sum_{u=1}^P \left(\sigma_u^0 Z_{u,i}^{'0} + \Delta_{u,i}^{0,0}\right)^2}\leq \sqrt{\sum_{u=1}^P \left(\sigma_u^0 Z_{u,i}^{'0}\right)^2} + \sqrt{\sum_{u=1}^P \left(\Delta_{u,i}^{0,0}\right)^2}$ (think of it as $\ell_2$-norm of length-$P$ vectors), followed by Cauchy-Schwarz and Lemma \ref{lem:squaredsums}.

From \eqref{eq:B0_b1}, we have
\begin{align*}
T_1 &\leq \sum_{u=1}^P P\left( \frac{\|Z_u^{'0}\|}{\sqrt{n}} - 1  \geq 1 \right) + P\left( \sum_{u=1}^P \frac{\|\Delta_u^{0,0}\|}{\sqrt{n}} \geq \frac{\e\min\{1,\frac{1}{2\sqrt{3}L}\}}{2+8\sum_{u=1}^P \sigma_u^0}\right)\\
&\overset{(a)}{\leq} \sum_{u=1}^P P\left( \frac{\|Z_u^{'0}\|}{\sqrt{n}} - 1  \geq 1 \right) + \sum_{u=1}^P P\left(\frac{\|\Delta_u^{0,0}\|}{\sqrt{n}} \geq \frac{\e\min\{1,\frac{1}{2\sqrt{3}L}\}}{P(2+8\sum_{u=1}^P \sigma_u^0)}\right)\\
&\overset{(b)}{\leq} K e^{-n} + K e^{-\kappa n\e},
\end{align*}
where step $(a)$ follows from Lemma \ref{sums} and step $(b)$ uses Lemma \ref{subexp} and $\mc{B}_0 (a)$.

Next consider $T_2$. Using Lemma \ref{sums}, we have
\begin{align*}
T_2  &\leq P\left( \abs{\frac{1}{n}\sum_{i=1}^n \left(\phi_b(\sigma_1^0 Z_{1,i}^{'0},...,\sigma_P^0 Z_{P,i}^{'0},w_i) - \mathbb{E}_{(Z_{1,i},...,Z_{P,i})}\left[\phi_b(\sigma_1^0 Z_{1,i}^{'0},...,\sigma_P^0 Z_{P,i}^{'0},w_i)\right]\right)} \geq \frac{\e}{4}  \right)\\
& + P\left( \abs{\frac{1}{n}\sum_{i=1}^n \mathbb{E}_{(Z_{1,i},...,Z_{P,i})}\left[\phi_b(\sigma_1^0 Z_{1,i}^{'0},...,\sigma_P^0 Z_{P,i}^{'0},w_i)\right] - \mathbb{E}\left[\phi_b(\sigma_1^0\breve{Z}_1^0,...,\sigma_P^0\breve{Z}_P^0,W)\right]} \geq \frac{\e}{4}  \right).
\end{align*}
Label the two terms above as $T_{2,a}$ and $T_{2,b}$.  $T_{2,a}$ has the desire bound by first noticing that the function $\tilde{\phi}_{b,i}:\mathbb{R}^P\rightarrow\mathbb{R}$ defined as $\tilde{\phi}_{b,i}(x_1,...,x_P):=\phi_b(x_1,...,x_P,w_i)$ is PL(2) since $\phi_{b}$ is, and then applying Lemma \ref{lem:subgauss_vec} to sequence of i.i.d. length-$P$ random vectors consisting of i.i.d. $\mc{N}(0,1)$ Gaussian random variables (variance factor $\nu=1$ in Lemma \ref{lem:subgauss_vec}). $T_{2,b}$ has the desired bound by first noticing that the function $\hat{\phi}_b:\mathbb{R}\rightarrow\mathbb{R}$ defined as $\hat{\phi}(x):=\mathbb{E}\left[\phi_b(\sigma_1^0\breve{Z}_1^0,...,\sigma_P^0\breve{Z}_P^0,x)\right]$ is PL(2) by Lemma \ref{lem:PLexamples}.

\textbf{(c)} The function $\phi_b(b_{1,i}^0,...,b_{P,i}^0,w_i):=b_{1,i}^0 w_i$ is PL(2) by Lemma \ref{lem:Lprods}. Apply $\mc{B}_0(b)(iii)$, we have
\begin{equation*}
\frac{1}{n}\sum_{i=1}^n b_{1,i}^0 w_i \doteq \mathbb{E}\left[\sigma_1^0\breve{Z}_1^0 W\right]=0.
\end{equation*}

\textbf{(d)} The function $\phi_b(b_{1,i}^0,...,b_{P,i}^0,w_i):=(b_{1,i}^0)^2$ is PL(2) by Lemma \ref{lem:Lprods}. Apply $\mc{B}_0(b)(iii)$, we have
\begin{equation*}
\frac{1}{n}\sum_{i=1}^n (b_{1,i}^0)^2 \doteq \mathbb{E}\left[(\sigma_1^0\breve{Z}_1^0)^2\right] = \tilde{E}_1^{0,0},
\end{equation*}
where the last equality uses \eqref{eq:def_Z_tilde}.

\textbf{(e)} Recall that $m_1^0 = b_1^0 - w$.
The function $\phi_b(b_{1,i}^0,...,b_{P,i}^0,w_i):=b_{1,i}^0 - w_i$ is PL(2). Apply $\mc{B}_0 (b)(iii)$, we have
\begin{equation*}
\frac{(m_1^0)^*m_1^0}{n} = \frac{1}{n}\sum_{i=1}^{n} (b_{1,i}^0 - w_i)^2 \doteq \mathbb{E}\left[(\sigma_1^0 \breve{Z}_1^0 - W)^2\right] = (\sigma_1^0)^2 + \sigma_W^2 = \breve{E}_1^{0,0}.
\end{equation*}
The last equality follows from \eqref{eq:def_tilde_E} by noticing that $\tilde{E}_1^{\theta(0),\theta(0)}=0$ by definition.

\textbf{(f)} The function $\phi_b(b_{1,i}^0,...,b_{P,i}^0,w_i):= b_{1,i}^0(b_{1,i}^0 - w_i)$ is PL(2) by Lemma \ref{lem:Lprods}. Apply $\mc{B}_0(b)(iii)$, we have
\begin{equation*}
\frac{(b_1^0)^*m_1^0}{n} = \frac{1}{n}b_{1,i}^0(b_{1,i}^0 - w_i) \doteq \mathbb{E}\left[ \sigma_1^0\breve{Z}_1^0( \sigma_1^0\breve{Z}_1^0 - W) \right] = \tilde{E}_1^{0,0}.
\end{equation*}

\textbf{(g)} Nothing to prove here.

\textbf{(h)} The result is equivalent to $\mc{B}_0(e)$, since $\|(m_p^0)_{\perp}\|=\|m_1^0\|$ and $(\tau_p^0)_{\perp}^2 = (\tau_p^0)^2$.

\subsection{Step 2: Showing $\mc{H}_{1}$ holds.}
The proof is the same as in \cite{RushV16} and is not repeated hear.

\subsection{Step 3: Showing $\mc{B}_t$ holds.} We prove the statements in $\mc{B}_t$ assuming that $\mc{B}_{t-1}$ and $\mc{H}_{t}$ holds by inductive hypothesis.

\textbf{(a)} The proof is the same as in \cite{RushV16}.

\textbf{(b)(iii)} For brevity, define $\mathbb{E}_{\phi_b}:=\mathbb{E}\left[\phi_b(\sigma_1^0\breve{Z}_1^0,...,\sigma_P^0\breve{Z}_P^0,...,\sigma_1^t\breve{Z}_1^t,...,\sigma_P^t\breve{Z}_P^t,W)\right]$, and
\begin{align*}
a_i &= (b_{1,i}^0,...,b_{P,i}^0,...,b_{1,i}^{t-1},...,b_{P,i}^{t-1},\sum_{r=0}^{t-1}\hat{\gamma}^t_{1,r}b_{1,i}^r+(\sigma_1^t)_{\perp}Z_{1,i}^{'t} + [\Delta_1^{t,t}]_i,...,\sum_{r=0}^{t-1}\hat{\gamma}^t_{P,r}b_{P,i}^r+(\sigma_P^t)_{\perp}Z_{P,i}^{'t} + [\Delta_P^{t,t}]_i,w_i)\\
c_i &= (b_{1,i}^0,...,b_{P,i}^0,...,b_{1,i}^{t-1},...,b_{P,i}^{t-1},\sum_{r=0}^{t-1}\hat{\gamma}^t_{1,r}b_{1,i}^r+(\sigma_1^t)_{\perp}Z_{1,i}^{'t},...,\sum_{r=0}^{t-1}\hat{\gamma}^t_{P,r}b_{P,i}^r+(\sigma_P^t)_{\perp}Z_{P,i}^{'t},w_i),
\end{align*}
for $i=1,...,n$. Hence, $a,c$ are length-$n$ vectors with elements in $\mathbb{R}^{P(t+1)+1}$.

Then, using the conditional distribution of $b_p^t$ defined in Lemma \ref{lem:cond_htbt} \eqref{eq:def_cond_bt}, together with Lemma \ref{sums}, we have
\begin{align*}
&P\left(\abs{\frac{1}{n}\sum_{i=1}^n \phi_b(b_{1,i}^0,...,b_{P,i}^0,...,b_{1,i}^{t},...,b_{P,i}^{t},w_i) - \mathbb{E}_{\phi_b}} \geq \e\right)\\
&\leq P\left(\abs{\frac{1}{n}\sum_{i=1}^n \left(\phi_b(a_i)-\phi_b(c_i)\right)}\geq \frac{\e}{2}\right) + P\left(\abs{\frac{1}{n}\sum_{i=1}^n \phi_b(c_i) - \mathbb{E}_{\phi_b}}\geq \frac{\e}{2}\right).
\end{align*}
Label the two terms above as $T_1$ and $T_2$. To complete the proof, we show that both $T_1$ and $T_2$ are bounded above by $Ke^{-\kappa n \e^2}$. First consider $T_1$. For convenience, we define $\underline{\Delta}^{t,t}_i:=([\Delta^{t,t}_1]_i,...,[\Delta^{t,t}_P]_i)$, for $i=1,...,n$. 
\begin{align*}
T_1 &\overset{(a)}{\leq} P\left( \frac{1}{n}\sum_{i=1}^n \left(1 + \|a_i\| + \|c_i\|\right)\|a_i-c_i\| \geq \frac{\e}{2L} \right)\overset{(b)}{\leq}P\left( \frac{1}{n}\sum_{i=1}^n \left(1 + \|\underline{\Delta}_i^{t,t}\| + 2\|c_i\|\right)\|\underline{\Delta}_i^{t,t}\| \geq \frac{\e}{2L} \right)\\
&\overset{(c)}{\leq} P\left( \sqrt{1+\sum_{u=1}^P \frac{\|\Delta_u^{t,t}\|^2}{n} + 4\frac{\|c\|^2}{n}} \sqrt{\sum_{u=1}^P \frac{\|\Delta_u^{t,t}\|^2}{n}} \geq \frac{\e}{2\sqrt{3}L}\right)\\
&\leq P\left( \left(1+\sum_{u=1}^P \frac{\|\Delta_u^{t,t}\|}{\sqrt{n}} + 2\frac{\|c\|}{\sqrt{n}} \right)\left(\sum_{u=1}^P \frac{\|\Delta_u^{t,t}\|}{\sqrt{n}}\right) \geq \frac{\e}{2\sqrt{3}L}\right),
\end{align*}
where step $(a)$ holds because $\phi_b$ is PL(2), step $(b)$ holds because $\|a_i\|\leq \|c_i\|+\|\underline{\Delta}_i^{t,t}\|$ and $\|a_i-c_i\| = \|\underline{\Delta}_i^{t,t}\|$, and step $(c)$ holds by Cauchy-Schwarz followed by Lemma \ref{lem:squaredsums}. Further, $\|c\|$ can be bounded as follows by applying Lemma \ref{lem:squaredsums},
\begin{align*}
\|c\|^2  \leq \sum_{u=1}^P \sum_{r=0}^{t-1} \|b_u^r\|^2 + 2\sum_{u=1}^P\sum_{r=0}^{t-1}\sum_{l=0}^{t-1}\hat{\gamma}^t_{u,r}\hat{\gamma}^t_{u,l} (b_u^r)^*b_u^l + 2\sum_{u=0}^P (\sigma_u^t)^2_{\perp} \|Z_u^{'t}\|^2 + \|w\|^2.
\end{align*}
Denote the RHS of the above by $\tilde{\mathsf{c}}$. We will show that $\frac{1}{n}\tilde{\mathsf{c}}$ concentrates to $\mathbb{E}_{\tilde{\mathsf{c}}}$ defined as
\begin{equation*}
\mathbb{E}_{\tilde{c}}:=\sum_{u=1}^P\sum_{r=0}^{t-1} \tilde{E}_u^{r,r} + 2\sum_{u=1}^{P}\sum_{r=0}^{t-1}\sum_{l=0}^{t-1} \hat{\gamma}^t_{u,r}\hat{\gamma}^t_{u,l} \tilde{E}_u^{r,l} + 2\sum_{u=1}^P (\sigma_u^t)^2_{\perp} + \sigma_W^2 = \sum_{u=1}^P \left((\sigma_u^l)^2 +2(\sigma_u^t)^2\right) + \sigma_W^2.
\end{equation*} 
In the above, the last inequality uses the definitions in Section \ref{subsec:conc_const} and follows from
\begin{equation}
\sum_{r=0}^{t-1}\sum_{l=0}^{t-1} \hat{\gamma}^t_{u,r}\hat{\gamma}^t_{u,l} \tilde{E}_u^{r,l} = (\hat{\gamma}^t_u)^*\tilde{C}_u^t\hat{\gamma}^t_u = [(\tilde{E}_u^t)^*(\tilde{C}_u^t)^{-1}]\tilde{C}_u^t [(\tilde{C}_u^t)^{-1}\tilde{E}_u^t] = (\tilde{E}_u^t)^*(\tilde{C}_u^t)^{-1}\tilde{E}_u^t = (\sigma_u^t)^2 - (\sigma_u^t)^2_\perp.
\label{eq:Btb3_double_sum}
\end{equation}
To see the concentration, let $\mathsf{a}=Pt+t^2+P+1$,
\begin{align}
&P\left(\abs{\frac{1}{n}\tilde{\mathsf{c}} - \mathbb{E}_{\tilde{\mathsf{c}}}} \geq \e'\right)\overset{(a)}{\leq} \sum_{u=1}^P\sum_{r=0}^{t-1} P\left(\abs{\frac{\|b_u^r\|^2}{n} - \tilde{E}_u^{r,r}} \geq \frac{\e'}{\mathsf{a}} \right) + \sum_{u=1}^P\sum_{r=0}^{t-1}\sum_{l=0}^{t-1} P\left( \abs{\frac{(b_u^r)^*b_u^l}{n} - \tilde{E}_u^{r,l}} \geq \frac{\e'}{2\mathsf{a}\hat{\gamma}^t_{u,r}\hat{\gamma}^t_{u,l}} \right)\nonumber\\
&+ P\left(\abs{\frac{\|Z_u^{'t}\|^2}{n} - 1} \geq \frac{\e'}{2\mathsf{a}(\sigma_u^t)_{\perp}^2}\right) + P\left(\abs{\frac{\|w\|^2}{n} - \sigma_W^2} \geq \frac{\e'}{\mathsf{a}}\right)\overset{(b)}{\leq} K e^{-\kappa n \e'^2},
\label{eq:Btb3_tildec}
\end{align}
where step $(a)$ follows from Lemma \ref{sums} and step $(b)$ follows from the inductive hypothesis $\mc{B}_0 (d)-\mc{B}_{t-1} (d)$, Lemma \ref{subexp}, and the assumption on $w$. Using this result, $T_1$ can be bounded as follows,
\begin{align*}
T_1 &\leq P\left( \left(1 + \sum_{u=1}^P \frac{\|\Delta_u^{t,t}\|}{\sqrt{n}} +  2\left(\frac{\tilde{\mathsf{c}}}{\sqrt{n}}-\mathbb{E}_{\tilde{\mathsf{c}}}\right)+2\mathbb{E}_{\tilde{\mathsf{c}}} \right) \left(\sum_{u=1}^P \frac{\|\Delta_u^{t,t}\|}{\sqrt{n}} \right)  \geq \frac{\e}{2\sqrt{3}L}\right)\\
&\overset{(a)}{\leq} P\left( \abs{\frac{\tilde{\mathsf{c}}}{\sqrt{n}} - \mathbb{E}_{\tilde{\mathsf{c}}}^{1/2}} \geq \e \right) + \sum_{u=1}^P P\left( \frac{\|\Delta_u^{t,t}\|}{\sqrt{n}} \geq \frac{\e}{2\sqrt{3}LP(4+2\mathbb{E}_{\tilde{\mathsf{c}}}^{1/2})}\right)\overset{(b)}{\leq} Ke^{-\kappa n \e^2} + Ke^{-\kappa n \e^2},
\end{align*}
where step $(a)$ uses Lemma \ref{sums} and step $(b)$ uses \eqref{eq:Btb3_tildec}, the results in $\mc{B}_t (a)$, and Lemma \ref{sqroots}.

Next consider $T_2$. Let $Z=(Z_1,...,Z_P)$ be a vector of i.i.d. $\mc{N}(0,1)$ random variables, and it is independent of $Z_u^{'r}$ for all $r=1,...,t-1$ and $u=1,...,P$. For each $i\in [n]$, define the function $\tilde{\phi_{b,i}}:\mathbb{R}^P\rightarrow\mathbb{R}$ as for $z=(z_1,...,z_P)\in\mathbb{R}^P$,
\begin{equation*}
\tilde{\phi}_{b,i}(z):=\phi_b(b_{1,i}^0,...,b_{P,i}^0,...,b_{1,i}^{t-1},...,b_{P,i}^{t-1},\sum_{r=0}^{t-1}\hat{\gamma}^t_{1,r}b_1^r + (\sigma_1^t)_{\perp}z_1,...,\sum_{r=0}^{t-1}\hat{\gamma}^t_{P,r}b_P^r + (\sigma_P^t)_{\perp}z_P,w_i).
\end{equation*}
Then $\tilde{\phi}_{b,i}$ is PL(2) for all $i\in[n]$.
For brevity, let $Z_i' := (Z_{1,i}^{'t},...,Z_{P,i}^{'t})$, for all $i\in [n]$, then
\begin{align*}
T_2 \leq P\left(\abs{\frac{1}{n}\sum_{i=1}^n \left(\tilde{\phi}_{b,i}(Z_i') - \mathbb{E}_Z [\tilde{\phi}_{b,i}(Z)\right)} \geq \frac{\e}{4}\right) + P\left(\abs{\frac{1}{n}\sum_{i=1}^n \mathbb{E}_Z [\tilde{\phi}_{b,i}(Z) - \mathbb{E}_{\phi_b}} \geq \frac{\e}{4}\right).
\end{align*}
Label the two terms in the above as $T_1$ and $T_2$. $T_1$ is bounded above by $Ke^{-\kappa n \e^2}$ using Lemma \ref{lem:subgauss_vec}. $T_2$ can be written as 
\begin{equation*}
T_2 = P\left( \abs{\frac{1}{n}\sum_{i=1}^n \phi_b'(b_{1,i}^0,...,b_{P,i}^0,...,b_{1,i}^{t-1},...,b_{P,i}^{t-1},w_i) - \mathbb{E}_{\phi_b}} \geq \frac{\e}{4} \right),
\end{equation*} 
where the function $\phi_b':\mathbb{R}^{P(t+1)+1}\rightarrow\mathbb{R}$ is defined as
\begin{align*}
&\phi_b'(b_{1,i}^0,...,b_{P,i}^0,...,b_{1,i}^{t-1},...,b_{P,i}^{t-1},w_i):=\mathbb{E}_Z[\tilde{\phi}_{b,i}(Z)]\\
&=\mathbb{E}_Z\left[ \phi_b(b_{1,i}^0,...,b_{P,i}^0,...,b_{1,i}^{t-1},...,b_{P,i}^{t-1},\sum_{r=0}^{t-1}\hat{\gamma}^t_{1,r}b_{1,i}^r+(\sigma_1^t)_{\perp}Z_1,...,\sum_{r=0}^{t-1}\hat{\gamma}^t_{P,r}b_{P,i}^r+(\sigma_P^t)_{\perp}Z_P,w_i) \right]
\end{align*}
Then $\phi_b'$ is PL(2) by Lemma \ref{lem:PLexamples}. By inductive hypothesis $\mc{B}_0 (b)(iii)-\mc{B}_{t-1}(b)(iii)$, we have
\begin{align*}
&\frac{1}{n}\sum_{i=1}^n \phi_b'(b_{1,i}^0,...,b_{P,i}^0,...,b_{1,i}^{t-1},...,b_{P,i}^{t-1},w_i) \doteq \mathbb{E}\left[\phi_b'(\sigma_1^0\breve{Z}_1^0,...,\sigma_P^0\breve{Z}_P^0,...,\sigma_1^{t-1}\breve{Z}_1^{t-1},...,\sigma_P^{t-1}\breve{Z}_P^{t-1},W)\right]\\
&=\mathbb{E}\left[ \phi_b(\sigma_1^0\breve{Z}_1^0,...,\sigma_P^0\breve{Z}_P^0,...,\sigma_1^{t-1}\breve{Z}_1^{t-1},...,\sigma_P^{t-1}\breve{Z}_P^{t-1},\sum_{r=0}^{t-1}\hat{\gamma}^t_{1,r}\sigma_1^{r}\breve{Z}_1^{r} + (\sigma_1^t)_{\perp}Z_1,...,\sum_{r=0}^{t-1}\hat{\gamma}^t_{1,r}\sigma_P^{r}\breve{Z}_P^{r}+ (\sigma_P^t)_{\perp}Z_P,W) \right].
\end{align*}
To obtain the desired result, we need to show
\begin{align}
&(\sigma_1^0\breve{Z}_1^0,...,\sigma_P^0\breve{Z}_P^0,...,\sigma_1^{t-1}\breve{Z}_1^{t-1},...,\sigma_P^{t-1}\breve{Z}_P^{t-1},\sum_{r=0}^{t-1}\hat{\gamma}^t_{1,r}\sigma_1^{r}\breve{Z}_1^{r}+ (\sigma_1^t)_{\perp}Z_1,...,\sum_{r=0}^{t-1}\hat{\gamma}^t_{1,r}\sigma_P^{r}\breve{Z}_P^{r}+ (\sigma_P^t)_{\perp}Z_P)\\
&\qquad\qquad\overset{d}{=} (\sigma_1^0\breve{Z}_1^0,...,\sigma_P^0\breve{Z}_P^0,...,\sigma_1^t\breve{Z}_1^t,...,\sigma_P^t\breve{Z}_P^t).
\label{eq:Btb3_equal_d}
\end{align}
Notice that these are zero-mean jointly Gaussian random vectors, hence we only need to demonstrate that there covariance matrices are equal. For the coordinates on the diagonal, we need to show
\begin{equation}
\mathbb{E}\left[ (\sum_{r=0}^{t-1}\hat{\gamma}^r_{p,r}\sigma_p^{r}\breve{Z}_p^{r}+ (\sigma_p^t)_{\perp}Z_p)^2 \right] = \mathbb{E}\left[ (\sigma_p^t\breve{Z}_p^t)^2 \right]=(\sigma_p^t)^2,\quad\forall p=1,...,P.
\end{equation}
This is true since
\begin{equation*}
\mathbb{E}\left[ (\sum_{r=0}^{t-1}\hat{\gamma}^r_{p,r}\sigma_p^{r}\breve{Z}_p^{r}+ (\sigma_p^t)_{\perp}Z_p)^2 \right] \overset{(a)}{=} \sum_{r=0}^{t-1}\sum_{l=0}^{t-1} \hat{\gamma}^t_{p,r} \hat{\gamma}^t_{p,l}\tilde{E}_p^{r,l} + (\sigma_p^t)_{\perp}^2\overset{(b)}{=}(\sigma_p^t)^2,
\end{equation*}
where step $(a)$ follows from the definition in Section \ref{subsec:conc_const} and step $(b)$ follows from \eqref{eq:Btb3_double_sum}. Next consider the off-diagonal terms. Notice that $\breve{Z}_p^r$ is independent of $\breve{Z}_q^l$ for all $0\leq r,l\leq t$ whenever $p\neq q$ and also the i.i.d. random vector $(Z_1,...Z_P)$ is independent of $\breve{Z}_p^r$ for all $1\leq p\leq P, 0\leq r\leq t$. Hence, we only need to show for all $0\leq l \leq t-1$, 
\begin{equation*}
\mathbb{E}\left[ (\sigma_p^l\breve{Z}_p^l)(\sum_{r=0}^{t-1}\hat{\gamma}^t_{p,r}\sigma_p^{r}\breve{Z}_p^{r}+ (\sigma_p^t)_{\perp}Z_p) \right] = \mathbb{E}\left[(\sigma_p^l\breve{Z}_p^l)(\sigma_p^t\breve{Z}_p^t)\right] = \tilde{E}_p^{l,t}.
\end{equation*}
This is true since
\begin{align*}
\mathbb{E}\left[ (\sigma_p^l\breve{Z}_p^l)(\sum_{r=0}^{t-1}\hat{\gamma}^t_{p,r}\sigma_p^{r}\breve{Z}_p^{r}+ (\sigma_p^t)_{\perp}Z_p) \right]\overset{(a)}{=}\sum_{r=0}^{t-1}\hat{\gamma}^t_{p,r}\tilde{E}_p^{l,t} \overset{(b)}{=} [\tilde{C}^t_p\hat{\gamma}^t_p]_{l+1} \overset{(c)}{=} \tilde{E}_p^{l,t},
\end{align*}
where step $(a)$, $(b)$, and $(c)$ uses definition for $\tilde{E}_p^{l,r}$, $\tilde{C}_p^t$, and $\hat{\gamma}^t_p$, respectively, in Section \ref{subsec:conc_const}. Hence, we have proved \eqref{eq:Btb3_equal_d}.

\textbf{(c)} Apply the result in $\mc{B}_t (b)(iii)$ to PL(2) function $\phi_b(b_{1,i}^0,...,b_{P,i}^0,...,b_{1,i}^t,...,b_{P,i}^t,w_i):=b_{1,i}^tw_i$, then we have
\begin{equation*}
\frac{1}{n}\sum_{i=1}^n b_{1,i}^tw_i \doteq \mathbb{E}\left[\sigma_1^t\breve{Z}_1^tW\right]=0.
\end{equation*}

\textbf{(d)} Apply the result in $\mc{B}_t (b)(iii)$ to PL(2) function $\phi_b(b_{1,i}^0,...,b_{P,i}^0,...,b_{1,i}^t,...,b_{P,i}^t,w_i):=b_{1,i}^r b_{1,i}^t$, then we have
\begin{equation*}
\frac{1}{n} \sum_{i=1}^n b_{1,i}^r b_{1,i}^t \doteq \sigma_1^r\sigma_1^t\mathbb{E}\left[\breve{Z}_1^r\breve{Z}_1^t\right] = \tilde{E}_1^{r,t}. 
\end{equation*}

\textbf{(e)} Define the function $g:\mathbb{R}^{P+1}\rightarrow\mathbb{R}$ as $g(x)=\sum_{i=1}^{P} x_i - x_{P+1}$ for $x=(x_1,...,x_{P+1})\in\mathbb{R}^{P+1}$. Notice that $g$ is  Lipschitz.  Recall the definition of $m_1^t$ in \eqref{eq:def_htbt}, we have $m_{1,i}^t = g(b_{2,i}^{\theta(t)},...,b_{P,i}^{\theta(t)},b_{1,i}^t,w_i)$. Notice that $\theta(t)$ equals to some integer that is less than or equal to $t$.  Therefore, the function defined as $\phi_b(b_{1,i}^0,...,b_{P,i}^0,...,b_{1,i}^t,...,b_{P,i}^t,w_i):=g(b_{2,i}^{\theta(r)},...,b_{P,i}^{\theta(r)},b_{1,i}^r,w_i)g(b_{2,i}^{\theta(t)},...,b_{P,i}^{\theta(t)},b_{1,i}^t,w_i)$ is PL(2) by Lemma \ref{lem:Lprods}. Applying the result in $\mc{B}_t (b)(iii)$, we have
\begin{align*}
\frac{1}{n}\sum_{i=1}^n m_{1,i}^r m_{1,i}^t &\doteq \mathbb{E}\left[ g(\sigma_2^{\theta(r)}\breve{Z}_2^{\theta(r)},...,\sigma_P^{\theta(r)}\breve{Z}_P^{\theta(r)},\sigma_1^r\breve{Z}_1^r,W)g(\sigma_2^{\theta(t)}\breve{Z}_2^{\theta(t)},...,\sigma_P^{\theta(t)}\breve{Z}_P^{\theta(t)},\sigma_1^t\breve{Z}_1^t,W) \right]\\
&=\tilde{E}_1^{r,t} + \sum_{u=2}^P \tilde{E}_u^{\theta(r),\theta(t)} + \sigma_W^2 = \breve{E}_1^{r,t}. 
\end{align*}

\textbf{(f)} Using the function $g$ defined above and applying $\mc{B}(b)(iii)$, we have
\begin{align*}
\frac{1}{n} \sum_{i=1}^n b_{1,i}^r m_{1,i}^t &= \frac{1}{n} \sum_{i=1}^n b_{1,i}^r g(b_{2,i}^{\theta(t)},...,b_{P,i}^{\theta(t)},b_{1,i}^t,w_i) \doteq \mathbb{E}\left[\sigma_1^r\breve{Z}_1^r g(\sigma_2^{\theta(t)}\breve{Z}_2^{\theta(t)},...,\sigma_P^{\theta(t)}\breve{Z}_P^{\theta(t)},\sigma_1^t\breve{Z}_1^t,W) \right]\\
 &= \mathbb{E}\left[\sigma_1^r\breve{Z}_1^r\left(\sum_{u=2}^P \sigma_u^{\theta(0)} \breve{Z}_u^{\theta(0)} + \sigma_1^t \breve{Z}_u^{t} - W\right)\right]=\tilde{E}_1^{r,t}.
\end{align*}
The proof for $\frac{1}{n} \sum_{i=1}^n b_{1,i}^t m_{1,i}^r \doteq  \tilde{E}_1^{r,t}$ is similar. 

The proof for \textbf{(g)} and \textbf{(h)} is the same as in \cite{RushV16}.

\subsection{Step 4: Showing $\mc{H}_{t+1}$ holds.}
The proof is the same as in \cite{RushV16} and is not repeated hear.

\appendices
\section{Sub-Gaussian Concentration Lemmas}

\begin{applem}\cite{RushV16}
\label{lem:normalconc}
For a standard Gaussian random variable $Z$ and  $\e > 0$,
$P\left( \abs{Z} \geq \e \right) \leq 2e^{-\frac{1}{2}\e^2}$.
\end{applem}

\begin{applem}\cite[$\chi^2$-concentration]{RushV16}
For  $Z_i$, $i \in [n]$ that are i.i.d. $\sim \mc{N}(0,1)$, and  $0 \leq \e \leq 1$,
\[P\left(\left \lvert \frac{1}{n}\sum_{i=1}^n Z_i^2 - 1\right \lvert \geq \e \right) \leq 2e^{-n \e^2/8}.\]
\label{subexp}
\end{applem}

%

\begin{applem}
\cite{BLMConc} Let $X$ be a centered sub-Gaussian random variable with variance factor $\nu$, i.e., $\ln \expec[e^{tX}] \leq \frac{t^2 \nu}{2}$, $\forall t \in \mathbb{R}$. Then $X$ satisfies:
\begin{enumerate}
\item For all $x> 0$, $P(X > x)  \vee P(X <-x) \leq e^{-\frac{x^2}{2\nu}}$, for all $x >0$.
\item For every integer $k \geq 1$,
\be
\expec[X^{2k}] \leq 2 (k !)(2 \nu)^k \leq (k !)(4 \nu)^k.
\label{eq:subgauss_moments}
\ee
\end{enumerate}
\label{lem:subgauss}
\end{applem}

\begin{applem}
\label{lem:subgauss_vec}
Let $\{Z_i\}_{i\in [n]}\in\mathbb{R}^d$ be a sequence of i.i.d. random vectors on $\mathbb{R}^d$, where $d$ is an integer and each $Z_i$ has i.i.d. sub-Gaussian entries with variance factor $\nu$. The functions $\{f_i\}_{i\in [n]}:\mathbb{R}^d \rightarrow \mathbb{R}$ are pseudo-Lipschitz of order 2 each with pseudo-Lipschitz constant $L_i$. Let $L:=\max_{i\in[n]}L_i$. Then for all $\e \in (0,1)$, there exist $K,\kappa>0$ depending on $d,L$ but not $n,\e$ such that
\begin{equation}
P\left( \abs{\frac{1}{n}\sum_{i=1}^n \left( f_i(Z_i) - \mathbb{E}\left[f_i(Z_i)\right] \right)} \geq \e \right)\leq K e^{-\kappa n \e^2}.
\end{equation}
\end{applem}
\begin{proof}
In what follows, we prove the upper-tail bound
\begin{equation}
P\left( \frac{1}{n}\sum_{i=1}^n \left( f_i(Z_i) - \mathbb{E}\left[f_i(Z_i)\right] \right) \geq \e \right)\leq K e^{-\kappa n \e^2}.
\label{eq:lem_subgauss_vec_1}
\end{equation}
The lower-tail bound can then be obtained similarly. Together they prove the desired result.

Using the Cram{\' e}r-Chernoff method:
\begin{equation}
P\left(\sum_{i=1}^n \left( f_i(Z_i) - \mathbb{E}\left[f_i(Z_i)\right] \right) \geq n\e \right)\leq e^{-nr\e}\mathbb{E}\left[\exp\left(r\sum_{i=1}^n \left( f_i(Z_i) - \mathbb{E}\left[f_i(Z_i)\right] \right)\right) \right],\quad \forall r>0.
\label{eq:lem_subgauss_vec_2}
\end{equation}

Next we will show that there exist $\kappa'>0$ such that the expectation in \eqref{eq:lem_subgauss_vec_2} is bounded by
\begin{equation}
\mathbb{E}\left[\exp\left(r\sum_{i=1}^n \left( f_i(Z_i) - \mathbb{E}\left[f_i(Z_i)\right] \right)\right) \right] \leq e^{\kappa' n r^2},\quad\text{for}\quad 0 < r \leq \left(5\sqrt{2}L(d\nu+12d^2\nu^2)^{1/2}\right)^{-1}.
\label{eq:lem_subgauss_vec_3}
\end{equation}
Then plugging \eqref{eq:lem_subgauss_vec_3} into \eqref{eq:lem_subgauss_vec_2}, and choosing the optimal $r=\frac{\e}{2\kappa'}$, we can ensure that $r$ falls within the effective region defined in \eqref{eq:lem_subgauss_vec_2} for all $\e \in (0,1)$ by choosing $\kappa'$ big enough.

Now we show \eqref{eq:lem_subgauss_vec_3}. Notice that 
\begin{equation}
\mathbb{E}\left[\exp\left(-r\sum_{i=1}^n \left( f_i(Z_i) - \mathbb{E}\left[f_i(Z_i)\right] \right)\right) \right] \overset{(a)}{\geq} \exp\left(-r\mathbb{E}\left[\sum_{i=1}^n \left( f_i(Z_i) - \mathbb{E}\left[f_i(Z_i)\right] \right)\right]\right) = 1,
\label{eq:lem_subgauss_vec_4}
\end{equation}
where step $(a)$ follows from Jensen's inequality. Let $\{\tilde{Z}_i\}_{i\in[n]}$ be an independent copy of $\{Z_i\}_{i\in [n]}$. Then we have
\begin{align}
&\mathbb{E}\left[\exp\left\{r\sum_{i=1}^n \left( f_i(Z_i) - \mathbb{E}\left[f_i(Z_i)\right] \right)\right\} \right]\cdot 1  \nonumber\\
&\overset{(a)}{\leq} \mathbb{E}\left[\exp\left\{r\sum_{i=1}^n \left( f_i(Z_i) - \mathbb{E}\left[f_i(Z_i)\right] \right)\right\} \right]\mathbb{E}\left[\exp\left\{-r\sum_{i=1}^n \left( f_i(\tilde{Z}_i) - \mathbb{E}\left[f_i(\tilde{Z}_i)\right] \right)\right\}\right] \nonumber\\
&\overset{(b)}{=} \prod_{i=1}^n \mathbb{E}\left[ \exp\left\{f_i(Z_i) - f_i(\tilde{Z}_i) \right\}  \right] = \prod_{i=1}^n \left(\sum_{k=0}^{\infty}\frac{1}{k!}\mathbb{E}\left[ \left(f_i(Z_i)-f_i(\tilde{Z}_i)\right)^k\right]\right)\nonumber\\
&\overset{(c)}{=} \prod_{i=1}^n \left(\sum_{k=0}^{\infty}\frac{r^{2k}}{(2k)!}\mathbb{E}\left[ \left(f_i(Z_i)-f_i(\tilde{Z}_i)\right)^{2k}\right]\right)
\label{eq:lem_subgauss_vec_5}
\end{align}
In the above, step $(a)$ follows from \eqref{eq:lem_subgauss_vec_4},step $(b)$ uses the fact that $\{\tilde{Z}_i\}$ is independent of $\{Z_i\}$ and also they each contain independent elements, and step $(c)$ holds since the odd order terms are zero. Now we bound the expectation term.
\begin{align}
\mathbb{E}\left[ \left(f_i(Z_i)-f_i(\tilde{Z}_i)\right)^{2k}\right]&\overset{(a)}{\leq} L_i^{2k} \mathbb{E}\left[\left((1+\|Z_i\|+\|\tilde{Z}_i\|)\|Z_i-\tilde{Z}_i\|\right)^{2k}\right]\nonumber\\
&\overset{(b)}{\leq}L^{2k}\mathbb{E}\left[\left(\|Z_i\|+\|\tilde{Z}_i\|+\|Z_i\|^2+\|\tilde{Z}_i\|^2+2\|Z_i\|\|\tilde{Z}_i\|\right)^{2k}\right]\nonumber\\
&\overset{(c)}{\leq} \frac{(5L)^{2k}}{5}\mathbb{E}\left[2\|Z_i\|^{2k}+2\|Z_i\|^{4k}+2^{2k}\|Z_i\|^{2k}\|\tilde{Z}_i\|^{2k}\right]\nonumber\\
&\overset{(d)}{\leq} \frac{(5L)^{2k}}{5} \left( 4(k)!(2d\nu)^{k} +  4(2k)!(2d\nu)^{2k} + 4(k!)^2(4d\nu)^{2k} \right),
\label{eq:lem_subgauss_vec_6}
\end{align}
where step $(a)$ holds due to the pseudo-Lipschitz property, step $(b)$ uses the fact that $L_i\leq L$ and  $\|Z_i-\tilde{Z}_i\|\leq \|Z_i\|+\|\tilde{Z}_i\|$, step $(c)$ uses Lemma \ref{lem:squaredsums}, and step $(d)$ uses Lemma \ref{lem:subgauss} and another application of Lemma \ref{lem:squaredsums}. Plugging \eqref{eq:lem_subgauss_vec_6} into \eqref{eq:lem_subgauss_vec_5}, we have
\begin{align*}
&\mathbb{E}\left[\exp\left\{r\sum_{i=1}^n \left( f_i(Z_i) - \mathbb{E}\left[f_i(Z_i)\right] \right)\right\} \right] \overset{(a)}{\leq} \frac{4}{5}\sum_{k=0}^{\infty} (25r^2L^2)^k \left( (d\nu)^k + (4d^2\nu^2)^k + (8d^2\nu^2)^k\right)\\
&\leq \sum_{k=0}^{\infty} (25L^2)^k \left( d\nu + 12d^2\nu^2\right)^k = \frac{1}{1-25r^2L^2\left( d\nu + 12d^2\nu^2\right)}\overset{(b)}{\leq} \exp\left(50r^2L^2\left( d\nu + 12d^2\nu^2\right) \right),
\end{align*}
where step $(a)$ uses the fact $2^k(k!)^2\leq (2k)!$ and step $(b)$ uses $(1-x)^{-1}\leq e^{2x}$ for $x\in[0,1/2]$, which leads to the effective region for $r$ to be $r\leq \left(5\sqrt{2}L(d\nu+12d^2\nu^2)^{1/2}\right)^{-1}$.
\end{proof}

\section{Algebraic Inequalities}
\begin{applem}\cite[Concentration of Sums]{RushV16}
\label{sums}
If random variables $X_1, \ldots, X_M$ satisfy $P(\abs{X_i} \geq \e) \leq e^{-n\kappa_i \e^2}$ for $1 \leq i \leq M$, then 
\ben
P\left(\abs{  \sum_{i=1}^M X_i } \geq \e\right) \leq \sum_{i=1}^M P\left(|X_i| \geq \frac{\e}{M}\right) \leq M e^{-n (\min_i \kappa_i) \e^2/M^2}.
\een
\end{applem}

\begin{applem}\cite[Concentration of Square Roots]{RushV16}
\label{sqroots}
Let $c \neq 0$. Then
\ben
\text{If } P\left(\left \lvert X_n^2 - c^2 \right \lvert \geq \epsilon \right) \leq e^{-\kappa n \epsilon^2},
\text{ then }
P \left(\left \lvert \abs{X_n} - \abs{c} \right \lvert \geq \epsilon \right) \leq e^{-\kappa n \abs{c}^2 \epsilon^2}.
\een
\end{applem}

\begin{applem}\cite[Concentration of Products]{RushV16}
\label{products} 
For random  variables $X,Y$ and non-zero constants $c_X, c_Y$, if
\ben
P\left( \left | X- c_X \right |  \geq \e \right) \leq K e^{-\kappa n \e^2}, \quad \text{ and } \quad P\left( \left | Y- c_Y \right |  \geq \e \right) \leq K e^{-\kappa n \e^2},
\een
then the probability  $P\left( \left | XY - c_Xc_Y \right |  \geq \e \right)$ is bounded by 
\begin{align*}  
&  P\left( \left | X- c_X \right |  \geq \min\left( \sqrt{\frac{\e}{3}}, \frac{\e}{3 c_Y} \right) \right)  +  
P\left( \left | Y- c_Y \right |  \geq \min\left( \sqrt{\frac{\e}{3}}, \frac{\e}{3 c_X} \right) \right) \\
& \qquad \qquad \leq 2K \exp\left\{-\frac{\kappa n \e^2}{9\max(1, c_X^2, c_Y^2)}\right\}.
\end{align*}
\end{applem}

\begin{applem}\cite[Concentration of Scalar Inverses]{RushV16}
\label{inverses} Assume $c \neq 0$ and $0<\e <1$. 
\ben
\text{If } P\left(\left \lvert X_n - c \right \lvert \geq \epsilon \right) \leq e^{-\kappa n \epsilon^2},
\text{ then }
P\left(\left \lvert X_n^{-1} - c^{-1} \right \lvert \geq \epsilon \right) \leq 2 e^{-n \kappa \e^2 c^2 \min\{c^2, 1\}/4}.
\een
\end{applem}

\begin{applem}\cite{RushV16}
For any scalars $a_1, ..., a_t$ and positive integer $m$, we have  $\left(\abs{a _1} + \ldots + \abs{a_t} \right)^m \leq t^{m-1} \sum_{i=1}^t \abs{a_i}^m$.
Consequently, for any vectors $\un{u}_1, \ldots, \un{u}_t \in \mathbb{R}^N$, $\norm{\sum_{k=1}^t \un{u}_k}^2 \leq t \sum_{k=1}^t \norm{\un{u}_k}^2$.
\label{lem:squaredsums}
\end{applem}

\section{Pseudo-Lipschitz Function Lemmas}

\begin{applem} \cite[Products of Lipschitz Functions are PL(2)]{RushV16} 
Let $f:\mathbb{R}^p \to \mathbb{R}$ and $g:\mathbb{R}^p \to \mathbb{R}$ be Lipschitz continuous.  Then the product function $h:\mathbb{R}^p \to \mathbb{R}$ defined as $h(x) := f(x) g(x)$ is pseudo-Lipschitz of order 2.
\label{lem:Lprods}
\end{applem}

\begin{applem} \cite{RushV16}
Let $\phi: \mathbb{R}^{t+2} \rightarrow \mathbb{R}$ be $PL(2)$. Let $(c_1, \ldots, c_{t+1})$ be constants.  The function $\tilde{\phi}: \mathbb{R}^{t+1} \rightarrow \mathbb{R}$ defined as
\begin{align}
\tilde{\phi}\left(v_1, \ldots, v_t,  w\right) &= \mathbb{E}_{Z}\left[\phi\left(v_1, \ldots, v_t, \sum_{r=1}^{t} c_{r} v_{r} + c_{t+1} Z, w\right)\right]
\label{eq:PLex2}
\end{align}
where $Z \sim \mc{N}(0,1)$, is then also PL(2).
\label{lem:PLexamples}
\end{applem}

\bibliographystyle{IEEEtran}
\bibliography{IEEEabrv,cites}

\begin{thebibliography}{10}
\providecommand{\url}[1]{#1}
\csname url@samestyle\endcsname
\providecommand{\newblock}{\relax}
\providecommand{\bibinfo}[2]{#2}
\providecommand{\BIBentrySTDinterwordspacing}{\spaceskip=0pt\relax}
\providecommand{\BIBentryALTinterwordstretchfactor}{4}
\providecommand{\BIBentryALTinterwordspacing}{\spaceskip=\fontdimen2\font plus
\BIBentryALTinterwordstretchfactor\fontdimen3\font minus
  \fontdimen4\font\relax}
\providecommand{\BIBforeignlanguage}[2]{{%
\expandafter\ifx\csname l@#1\endcsname\relax
\typeout{** WARNING: IEEEtran.bst: No hyphenation pattern has been}%
\typeout{** loaded for the language `#1'. Using the pattern for}%
\typeout{** the default language instead.}%
\else
\language=\csname l@#1\endcsname
\fi
#2}}
\providecommand{\BIBdecl}{\relax}
\BIBdecl

\bibitem{MaLuBaron2017}
Y.~Ma, M.~Y. Lu, and D.~Baron, ``Multiprocessor approximate message passing
  \\with column-wise partitioning,'' in \emph{IEEE Int. Conf. Acoustics,
  Speech, Signal Process. (ICASSP)}, Mar. 2017.

\bibitem{Hastie2001}
T.~Hastie, R.~Tibshirani, and J.~H. Friedman, \emph{The Elements of Statistical
  Learning}.\hskip 1em plus 0.5em minus 0.4em\relax Springer, Aug. 2001.

\bibitem{Zhou2014}
Y.~Zhou, U.~Porwal, C.~Zhang, H.~Ngo, L.~Nguyen, C.~R{\'e}, and V.~Govindaraju,
  ``Parallel feature selection inspired by group testing,'' in \emph{Neural
  Inf. Process. Syst. (NIPS)}, Dec. 2014, pp. 3554--3562.

\bibitem{WangDunsonLeng2016}
X.~Wang, D.~Dunson, and C.~Leng, ``{DECO}rrelated feature space partitioning
  for distributed sparse regression,'' \emph{Arxiv preprint arXiv:1602.02575},
  Feb. 2016.

\bibitem{PengYanYin2013}
Z.~Peng, M.~Yan, and W.~Yin, ``Parallel and distributed sparse optimization,''
  in \emph{Proc. IEEE 47th Asilomar Conf. Signals, Syst., and Comput.}, Nov.
  2013, pp. 659--646.

\bibitem{Beck2009FISTA}
A.~Beck and M.~Teboulle, ``A fast iterative shrinkage-thresholding algorithm
  for linear inverse problems,'' \emph{SIAM J. Imag. Sci.}, vol.~2, no.~1, pp.
  183--202, Mar. 2009.

\bibitem{DMM2009}
D.~L. Donoho, A.~Maleki, and A.~Montanari, ``{Message passing algorithms for
  compressed sensing},'' \emph{Proc. Nat. Academy Sci.}, vol. 106, no.~45, pp.
  18\,914--18\,919, Nov. 2009.

\bibitem{Bayati2011}
M.~Bayati and A.~Montanari, ``The dynamics of message passing on dense graphs,
  with applications to compressed sensing,'' \emph{IEEE Trans. Inf. Theory},
  vol.~57, no.~2, pp. 764--785, Feb. 2011.

\bibitem{RushV16}
C.~Rush and R.~Venkataramanan, ``Finite sample analysis of approximate message
  passing,'' \emph{Proc. IEEE Int. Symp. Inf. Theory}, June 2015, full version:
  \url{https://arxiv.org/abs/1606.01800}.

\bibitem{Han2014}
P.~Han, R.~Niu, M.~Ren, and Y.~C. Eldar, ``Distributed approximate message
  passing for sparse signal recovery,'' in \emph{Proc. IEEE Global Conf. Signal
  Inf. Process.}, Atlanta, GA, Dec. 2014, pp. 497--501.

\bibitem{Han2015ICASSP}
P.~Han, R.~Niu, and Y.~C. Eldar, ``Modified distributed iterative hard
  thresholding,'' in \emph{IEEE Int. Conf. Acoustics, Speech, Signal Process.
  (ICASSP)}, Brisbane, Australia, Apr. 2015, pp. 3766--3770.

\bibitem{HanZhuNiuBaron2016ICASSP}
P.~Han, J.~Zhu, R.~Niu, and D.~Baron, ``Multi-processor approximate message
  passing using lossy compression,'' in \emph{IEEE Int. Conf. Acoustics,
  Speech, Signal Process. (ICASSP)}, Shanghai, China, Mar. 2016.

\bibitem{ZhuBeiramiBaron2016ISIT}
J.~Zhu, A.~Beirami, and D.~Baron, ``Performance trade-offs in multi-processor
  approximate message passing,'' in \emph{Proc. IEEE Int. Symp. Inf. Theory
  (ISIT)}, Barcelona, Spain, July 2016.

\bibitem{ZhuBaronMPAMP2016ArXiv}
J.~Zhu and D.~Baron, ``Multi-processor approximate message passing with lossy
  compression,'' \emph{Arxiv preprint arXiv:1601.03790}, Jan. 2016.

\bibitem{Krzakala2012probabilistic}
F.~Krzakala, M.~M{\'e}zard, F.~Sausset, Y.~Sun, and L.~Zdeborov{\'a},
  ``Probabilistic reconstruction in compressed sensing: {A}lgorithms, phase
  diagrams, and threshold achieving matrices,'' \emph{J. Stat. Mech. - Theory
  E.}, vol. 2012, no.~08, p. P08009, Aug. 2012.

\bibitem{BLMConc}
S.~Boucheron, G.~Lugosi, and P.~Massart, \emph{Concentration inequalities: A
  nonasymptotic theory of independence}.\hskip 1em plus 0.5em minus 0.4em\relax
  OUP Oxford, 2013.

\end{thebibliography}
\end{document}